\pdfoutput=1 
\documentclass[a4paper,UKenglish,cleveref, autoref, thm-restate]{lipics-v2019}
\usepackage[subtle, mathspacing=normal]{savetrees} 
\usepackage{csquotes}
\usepackage{xspace}
\usepackage{cite}
\newcommand{\defn}[1]{{\textit{\textbf{\boldmath #1}}}\xspace}
\renewcommand{\epsilon}{\varepsilon}

\crefname{equation}{}{} 
\crefname{enumi}{Step}{} 

\usepackage{algorithm}
\usepackage[noend]{algpseudocode} 

\title{Incremental Edge Orientation in Forests}

\titlerunning{} 

\author{Michael A. Bender}{Stony Brook University}{bender@cs.stonybrook.edu}{}{}
\author{Tsvi Kopelowitz}{
Bar-Ilan University
}{kopelot@gmail.com}{}{}
\author{William Kuszmaul}{Massachusetts Institute of Technology}{kuszmaul@mit.edu}{}{}
\author{Ely Porat}{Bar-Ilan University}{porately@cs.biu.ac.il}{}{}
\author{Clifford Stein}{Columbia University}{cliff@ieor.columbia.edu}{}{}

\authorrunning{}

\Copyright{}

\ccsdesc[300]{Theory of computation~Data structures design and analysis}
\ccsdesc[300]{Theory of computation~Dynamic graph algorithms}
\keywords{edge orientation, graph algorithms, Cuckoo hashing, hash functions}

\category{} 

\relatedversion{} 

\supplement{}

\funding{Supported in part by ISF grants no. 1278/16 and 1926/19, by a BSF grant no. 2018364, and by an ERC grant
MPM under the EU’s Horizon 2020 Research and Innovation Programme (grant no. 683064). This research was sponsored in part by National Science Foundation Grants XPS-1533644 and CCF-1930579, CCF-1714818, and CCF-1822809, CCF-2106827, 
CCF-1725543, 
CSR-1763680,  
CCF-1716252, and 
CNS-1938709. The research was also sponsored in part by the United States Air Force Research Laboratory and was accomplished under Cooperative Agreement Number FA8750-19-2-1000.  The views and conclusions contained in this document are those of the authors and should not be interpreted as representing the official policies, either expressed or implied, of the United States Air Force or the U.S. Government.  The U.S. Government is authorized to reproduce and distribute reprints for Government purposes notwithstanding any copyright notation herein.}

\acknowledgements{}

\nolinenumbers 

\newcommand{\punt}[1]{}

\newcommand{\E}{\mathbb{E}}

\makeatletter
\makeatletter
\def\@copyrightspace{\relax}
\makeatother

\renewcommand{\defn}[1]       {{\textit{\textbf{\boldmath #1}}}}
\newcommand{\pparagraph}[1]{\vspace{0.09in}\noindent{\bf \boldmath #1}}
\renewcommand{\paragraph}[1]{\vspace{0.09in}\noindent{\bf \boldmath #1.}} 
\newcommand{\poly}{\mbox{poly}}

\date{}

\renewcommand{\epsilon}{\varepsilon}

\newcommand{\seclabel}[1]    {\label{sec:#1}}

\begin{document}

\maketitle




\sloppy


  \sloppy
  \begin{abstract}
    For any forest $G = (V, E)$ it is possible to orient the edges $E$ so that no vertex in $V$ has out-degree greater than $1$. This paper considers the incremental edge-orientation problem, in which the edges $E$ arrive over time and the algorithm must maintain a low-out-degree edge orientation at all times. We give an algorithm that maintains a maximum out-degree of $3$ while flipping at most $O(\log \log n)$ edge orientations per edge insertion, with high probability in $n$. The algorithm requires worst-case time $O(\log n \log \log n)$ per insertion, and takes amortized time $O(1)$. The previous state of the art required up to $O(\log n / \log \log n)$ edge flips per insertion.

       We then apply our edge-orientation results to the problem of
    dynamic Cuckoo hashing. The problem of designing simple families
    $\mathcal{H}$ of hash functions that are compatible with Cuckoo
    hashing has received extensive attention. These families
    $\mathcal{H}$ are known to satisfy \emph{static guarantees}, but
    do not come typically with \emph{dynamic guarantees} for the running time of
    inserts and deletes. We show how to transform static guarantees
    (for $1$-associativity) into near-state-of-the-art dynamic
    guarantees (for $O(1)$-associativity) in a black-box
    fashion. Rather than relying on the family $\mathcal{H}$ to supply
    randomness, as in past work, we instead rely on randomness within
    our table-maintenance algorithm.

\end{abstract}



\thispagestyle{empty}

\newpage

\setcounter{page}{1}

\section{Introduction}
\seclabel{intro}


The general problem of maintaining low-out-degree edge orientations of
graphs has been widely studied and has found a broad
range of applications throughout algorithms (see, e.g., work on sparse graph representations \cite{BF99},  maximal matchings \cite{NeimanS13, HTZ14, KKPS14, BB17, BS15, BS16}, dynamic matrix-by-vector multiplication \cite{KKPS14}, etc.). However, some of the most basic and fundamental versions of the graph-orientation problem have remained unanswered.

This paper considers the problem of incremental edge
orientation in forests. Consider a sequence of edges
$e_1, e_2, \ldots, e_{n - 1}$ that arrive over time, collectively
forming a tree. As the edges arrive, one must maintain an
\defn{orientation} of the edges (i.e., to assign a direction to each
edge) so that no vertex ever has out-degree greater than $O(1)$. The
orientation can be updated over time, meaning that orientations of old
edges can be flipped in order to make room for the newly inserted
edges. The goal is achieve out-degree $O(1)$ while flipping as few
edges as possible per new edge arrival.

Forests represent the best possible case for edge orientation: it 
is always possible to construct an orientation with maximum out-degree
$1$. But, even in this seemingly simple case, no algorithms are known 
that achieve better than $O(\log  / \log \log n)$ edge flips per edge 
insertion \cite{KKPS14}. A central result of this paper is that, 
by using randomized and intentionally non-greedy edge-flipping one can 
can do substantially better, achieving $O(\log \log n)$ edges flips per insertion.

\pparagraph{A warmup: two simple algorithms.}
As a warmup let us consider two simple algorithms for incremental
edge-orientation in forests.

The first algorithm never flips any edges but allows the maximum out-degree of each vertex to be as high as $O(\log n)$. When an edge $(u, v)$ is added to the graph, the algorithm examines the connected components $T_u$ and $T_v$ that are being connected by the edge, and determines which component is larger (say, $|T_v| \ge |T_u|$). The algorithm then orients the edge from $u$ to $v$, so that it is directed out of the smaller component. Since the new edge is always added to a vertex whose connected component at least doubles in size, the maximum out-degree is $\lceil \log n \rceil$.

The second algorithm guarantees that the out-degree will always be 1, but at the cost of flipping more edges. As before, when $(u,v)$ is added the algorithm orients the edge from $u$ to $v$. If this increments the out-degree of $u$ to $2$, then the algorithm follows the directed path $P$ in $T_u$ starting from $u$ (and such that the edge $(u,v)$ is not part of $P$) until a vertex $r$ with out-degree 0 is reached.  The algorithm then flips the edge orientations on $P$, which increases the out-degree of $r$ to be $1$ and reduces the out-degree of $u$ to be $1$. Since every edge that is flipped is always part of a connected component that has just at least doubled in size, the number of times each edge is flipped (in total across all insertions) is at most $\lceil \log n \rceil$ and so the amortized time cost per insertion is $O(\log n)$.\footnote{By allowing for a maximum out-degree of $2$, the bound of $O(\log n)$ on the number of edges flipped can be improved from being amortized to worst-case. In particular, for any vertex $v$ there is always a (directed) path of length $O(\log n)$ to another vertex with out-degree $1$ or less (going through vertices with out-degree $2$); by flipping the edges in such a path, we can insert a new edge at the cost of only $O(\log n)$ flips.}

These two algorithms sit on opposite sides of a tradeoff curve. In one case, we have maximum out-degree $O(\log n)$ and at most $O(1)$ edges flipped per insertion, and in the other we have maximum out-degree $O(1)$ and at most $O(\log n)$ (amortized) flips per insertion. This raises a natural question: \emph{what is the optimal tradeoff curve between the maximum out-degree and the number of edges flipped per insertion?} 

\pparagraph{Our results.}
We present an algorithm for incremental edge orientation in forests that satisfies the following guarantees with high probability in $n$:
\begin{itemize}
\item the maximum out-degree never exceeds $3$;
\item the maximum number of edges flipped per insertion is $O(\log \log n)$; 
\item the maximum time taken by any insertion is $O(\log n \log \log n)$;
\item and the amortized time taken (and thus also the amortized number of edges flipped) per insertion is $O(1)$. 
\end{itemize}

An interesting feature of this result is that the aforementioned tradeoff curve is actually quite different than it first seems: by increasing the maximum out-degree to $3$ (instead of $2$ or $1$), we can decrease the maximum number of edges flipped per insertion all the way to $O(\log \log n)$.

In fact, a similar phenomenon happens on the other side of the tradeoff curve. For any $\epsilon$, we show that it is possible to achieve a maximum out-degree of $\log^{\epsilon} n + 1$ while only flipping $O(\epsilon^{-1})$ edges per insertion. Notably, this means that, for any positive constant $c$, one can can achieve out-degree $(\log n)^{1/ c}$ with $O(1)$ edges flipped per insertion. 

A key idea in achieving the guarantees above is to selectively leave vertices with low out-degrees ``sprinkled'' around the graph, thereby achieving an edge orientation that is amenable to future edges being added. Algorithmically, the main problem that our algorithm solves is that of high-degree vertices clustering in a ``hotspot'', which could then force a single edge-insertion to invoke a large number of edge flips. 


\pparagraph{Related work on edge orientations.}
The general problem of maintaining low-out-degree orientations of
dynamic graphs has served as a fundamental tool for many
problems. Brodal and Fagerberg \cite{BF99} used low-degree edge
orientations to represent dynamic sparse graphs -- by assigning each
vertex only $O(1)$ edges for which it is responsible, one can then
deterministically answer adjacency queries in $O(1)$ time. Low-degree
edge orientations have also been used to maintain maximal matchings in
dynamic graphs~\cite{NeimanS13, HTZ14, KKPS14, BB17}, and this
technique remains the state of the art for graphs with low
arboricity. Other applications include dynamic matrix-by-vector
multiplication~\cite{KKPS14}, dynamic shortest-path queries in planar
graphs~\cite{KK06}, and approximate dynamic maximum
matchings~\cite{BS15,BS16}. 

The minimum out-degree attainable by any orientation of a graph is
determined by the graph's pseudo-arboricity $\alpha$. As a result, the
algorithmic usefulness of low out-degree orientations is most
significant for graphs that have low pseudo-arboricity.  This makes
forests and pseudoforests (which are forests with one extra edge per
component) especially interesting, since they represent the case of
$\alpha = 1$ and thus always allow for an orientation with out-degree~$1$. 

Whereas this paper focuses on edge orientation in incremental forests
(and thus also incremental pseudoforests), past work has considered a
slightly more general problem~\cite{BF99, KKPS14, BB17, HTZ14},
allowing for edge deletions in addition to edge insertions, and also
considering dynamic graphs with pseudo-arboricities $\alpha > 1$. 
Brodal and Fagerberg gave an algorithm that achieved out-degree $O(\alpha)$ with amortized running time that is guaranteed to be constant competitive with that of any algorithm; they also showed that in the case of $\alpha \in O(1)$, it is possible to achieve constant out-degree with amortized time $O(1)$ per insertion and $O(\log n)$ per deletion~\cite{BF99}. For worst-case guarantees, on the other hand,  the only algorithm known to achieve sub-logarithmic bounds for both out-degree and edges flipped per insertion is that of Kopelowitz et al.~\cite{KKPS14}, which achieves $O(\log n / \log \log n)$ for both, assuming $\alpha \in O(\sqrt{\log n})$. In the case of incremental forests, our results allow for us to improve substantially on this, achieving a worst-case bound of $O(\log \log n)$ edges flipped per insertion (with high probability) while supporting maximum out-degree $O(1)$. An interesting feature of our algorithm is that it is substantially 
different than any of the past algorithms, suggesting that
the fully dynamic graph setting (with $\alpha > 1$) may warrant revisiting. 

Our interest in the incremental forest case stems in part from its importance
for a specific application: Cuckoo hashing. As we shall now discuss,
our results on incremental edge orientation immediately yield a
somewhat surprising result on Cuckoo hashing with dynamic guarantees.

\subsection{An Application to Cuckoo Hashing:  From Static to Dynamic Guarantees via Non-Greedy Eviction}

A \defn{$s$-associative Cuckoo hash table} \cite{pagh2001cuckoo, devroye2003cuckoo, li2014algorithmic,panigrahy2005efficient} consists of $n$ \defn{bins}, each of which has $s$ \defn{slots}, where $s$ is a constant typically between $1$ and $8$ \cite{pagh2001cuckoo, li2014algorithmic}. Records are inserted into the table using two hash functions $h_1, h_2$, each of which maps records to bins. The invariant that makes Cuckoo hashing special is that, if a record $x$ is in the table, then $x$ must reside in either bin $h_1(x)$ or $h_2(x)$. This invariant ensures that query operations \emph{deterministically} run in time $O(1)$.

When a new record $x$ is inserted into the table, there may not initially be room in either bin $h_1(x)$ or $h_2(x)$. In this case, $x$ \defn{kicks out} some record $y_1$ in either $h_1(x)$ or $h_2(x)$. This, in turn, forces $y_1$ to be inserted into the other bin $b_2$ to which $y_1$ is hashed. If bin $b_2$ \emph{also} does not have space, then $y_1$ kicks out some record $y_2$ from bin $b_2$, and so on. This causes what is known as a \defn{kickout chain}. Formally, a kickout chain takes a sequence of records $y_1, y_2, \ldots, y_j$ that reside in bins $b_1, b_2, \ldots, b_j$, respectively, and relocates those records to instead reside in bins $b_2, b_3, \ldots, b_{j + 1}$, respectively, where for each record $y_i$ the bins $b_i$ and $b_{i + 1}$ are the two bins to which $h_1$ and $h_2$ map $y_i$. The purpose of a kickout chain is to free up a slot in bin $b_1$ so that the newly inserted record can reside there. Although Cuckoo hashing guarantees constant-time queries, insertion operations can sometimes incur high latency due to long kickout chains.

The problem of designing simple hash-function families for Cuckoo hashing has received extensive attention \cite{mitzenmacher2008simple, patracscu2012power,aamand2018power,dietzfelbinger2009risks, aumuller2014explicit,dietzfelbinger2003almost,aumuller2016simple,pagh2001cuckoo, cohen2009bounds}. Several natural (and widely used) families of hash functions are known \emph{not} to work \cite{dietzfelbinger2009risks,cohen2009bounds}, and it remains open whether there exists $k = o(\log n)$ for which $k$-independence suffices \cite{mitzenmacher2009some}. This has led researchers to design and analyze \emph{specific families} of simple hash functions that have low independence but that, nonetheless, work well with Cuckoo hashing \cite{mitzenmacher2008simple, patracscu2012power,aamand2018power, aumuller2014explicit,dietzfelbinger2003almost,aumuller2016simple,pagh2001cuckoo}. Notably, Cuckoo hashing has served as one of the main motivations for the intensive study of tabulation hash functions \cite{aamand2018power, patracscu2012power, puatracscu2013twisted,thorup2017fast,dahlgaard2014approximately}.

Work on hash-function families for cuckoo hashing \cite{mitzenmacher2008simple, patracscu2012power,aamand2018power, aumuller2014explicit,dietzfelbinger2003almost,aumuller2016simple,pagh2001cuckoo} has focused on offering a \defn{static guarantee}: for any set $X$ of $O(n)$ records, there exists (with reasonably high probability) a valid 1-associative hash-table configuration that stores the records $X$. This guarantee is static in the sense that it does not say anything about the speed with which insertion and deletion operations can be performed. 


On the other hand, if the hash functions are fully random, then a strong \defn{dynamic guarantee} is known. Panigrahy \cite{panigrahy2005efficient} showed that, using bins of size two, insertions can be implemented to incur at most $\log \log n + O(1)$ kickouts, and to run in time at most $O(\log n)$, with high probability in $n$. Moreover, the expected time taken by each insertion is $O(1)$.

The use of bin sizes greater than one is essential here, as it gives the data structure algorithmic flexibility in choosing which record to evict from a bin. Panigrahy \cite{panigrahy2005efficient} uses breadth-first search in order to find the shortest possible kickout chain to a bin with a free slot. The fact that the hash functions $h_1$ and $h_2$ are fully random ensures that, with high probability, the search terminates within $O(\log n)$ steps, thereby finding a kickout chain of length $\log \log n + O(1)$.

If a family of hash functions has sufficiently strong randomness properties (e.g., the family of \cite{dietzfelbinger2003almost}) 
then one can likely recreate the guarantees of \cite{panigrahy2005efficient} by directly replicating the analysis. 
For other families of hash functions \cite{mitzenmacher2008simple, patracscu2012power,aamand2018power, aumuller2014explicit,dietzfelbinger2003almost,aumuller2016simple,pagh2001cuckoo}, however, it is unclear what sort of dynamic guarantees are or are not possible.

This raises a natural question: \emph{does there exist a similar dynamic guarantee to that of \cite{panigrahy2005efficient} when the underlying hash functions are not fully random -- in particular, if we know only that a hash family $\mathcal{H}$ offers a static guarantee, but we know nothing else about the structure or behavior of hash functions in $\mathcal{H}$, is it possible to transform the static guarantee into a dynamic guarantee? }

\pparagraph{Our results on Cuckoo hashing: a static-to-dynamic transformation.} We answer this question in the affirmative by presenting a new algorithm, the Dancing-Kickout Algorithm, for selecting kickout chains during insertions in a Cuckoo hash table. Given any hash family $\mathcal{H}$ that offers a $1$-associative static guarantee, we show that the same hash family can be used to offer an $O(1)$-associative dynamic guarantee. In particular, the Dancing-Kickout Algorithm supports both insertions and deletions with the following promise: as long as the static guarantee for $\mathcal{H}$ has not failed, then with high probability, each insertion/deletion  incurs at most $O(\log \log n)$ kickouts, has amortized time (and therefore number of kickouts) $O(1)$, and takes time at most $O(\log n \log \log n)$. We also extend our results to consider families of hash functions $\mathcal{H}$ that offer relaxed static guarantees -- that is, our results still apply to families either make assumptions about the input set \cite{mitzenmacher2008simple} or require the use of a small auxiliary stash~\cite{aumuller2014explicit, kirsch2010more}. 

Unlike prior algorithms, the Dancing-Kickout Algorithm takes a \emph{non-greedy} approach to record-eviction. The algorithm will sometimes continue a kickout chain \emph{past} a bin that has a free slot, in order to avoid ``hotspot clusters'' of full bins within the hash table. These hotspots are avoided by ensuring that, whenever a bin surrenders its final free slot, the bin is at the end of a reasonably long random walk, and is thus itself a ``reasonably'' random bin. Intuitively, the random structure that the algorithm instills into the hash table makes it possible for the hash functions from $\mathcal{H}$ to not be fully random.

The problem of low-latency Cuckoo hashing is closely related to the problem of incremental edge orientation. In particular, the static guarantee for a Cuckoo hash table (with bins of size one) means that the edges in a certain graph form a pseudoforest. And the problem of dynamically maintaining a Cuckoo hash table (with bins of size $O(1)$) can be solved by dynamically orienting the pseudoforest in order to maintain constant out-degrees. The Dancing-Kickout algorithm is derived by applying our results for incremental edge orientation along with several additional ideas to handle deletions.

In addition to maintaining $n$ bins, the Dancing-Kickout Algorithm uses an auxiliary data structure of size $O(n)$. The data structure incurs at most $O(1)$ modifications per insertion/deletion. Importantly, the auxiliary data structure is not accessed during queries, which continue to be implemented as in a standard Cuckoo hash table.

Our results come with an interesting lesson regarding the symbiotic relationship between Cuckoo hashing and edge orientation. There has been a great deal of past work on Cuckoo hashing that focuses on parameters such as associativity, number of hash functions, and choice of hash function. 
We show that a new dimension that also warrants attention:  how to dynamically maintain the table to ensure that a short kickout chain exists for every insertion. 
Algorithms that greedily optimize \emph{any given operation} (e.g., random walk and BFS) may inadvertently structure the table in a way that compromises the performance of some later operations. In contrast, the non-greedy approach explored in this paper is able to offer strong performance guarantees for all operations, even if the hash functions being used are far from fully random. 
The results in this paper apply only to $1$-associative static guarantees, and are therefore innately limited in the types of dynamic guarantees that they can offer (for example, we cannot hope to support a load factor of better than $0.5$). An appealing direction for future work is to design and analyze eviction algorithms that offer strong dynamic guarantees in hash tables with either a large associativity or a large number of hash functions---it would be especially interesting if such guarantees could be used to support a load factor of $1 - q$ for an arbitrarily small positive constant $q$.


\pparagraph{Related work on low-latency hash tables.}
Several papers have used ideas from Cuckoo hashing as a parts of new data structures that achieve stronger guarantees. Arbitman et al. \cite{deamortized1} showed how to achieve a fully constant-time hash table by maintaining a polylogarithmic-size backyard consisting of the elements whose insertions have not yet completed at any given moment. Subsequent work then showed that, by storing almost all elements in a balls-in-bins system and then storing only a few ``overflow'' elements in a backyard Cuckoo hash table, one can construct a succinct constant-time hash table \cite{succinct}.\footnote{It is worth noting, however, that as discussed in \cite{succinct2}, the data structure of \cite{succinct} can be modified to use any constant-time hash table in place of deamortized Cuckoo hashing.}

Whereas the focus of these papers \cite{deamortized1, succinct} is to design new data structures that build on top of Cuckoo hashing, the purpose of our results is to consider \emph{standard} Cuckoo hashing but in the dynamic setting. In particular, our goal is to show that dynamic guarantees for Cuckoo hashing do not have to be restricted to fully random hash-functions; by using the Dancing-Kickout Algorithm for maintaining the Cuckoo hash table, \emph{any} family of hash functions that enjoys static guarantees can also enjoy dynamic guarantees. 


\subsection{Outline}
The paper proceeds as follows. In Section \ref{sec:technical}, we give
a technical overview of the algorithms and analyses in this paper. 
The overview is written in a way so that all of the major ideas in the paper 
are self contained. The
full details of the analyses are then given in appendices. Appendix
\ref{sec:alg} shows how to achieve $O(1)$ out-degree with
$O(\log \log n)$ edge flips per edge insertion; Appendix
\ref{sec:amortized} shows how to optimize the running time to be
$O(\log n \log \log n)$ per operation and $O(n)$ in total; Appendix
\ref{sec:tradeoff} then considers the tradeoff curve between out-degree
and number of edges flipped per insertion; finally, Appendix
\ref{sec:cuckoo} gives the full details of our application to Cuckoo
hashing.


\section{Technical Overview}
\label{sec:technical}

This section overviews the main technical ideas in the paper. We first describe our results for incremental edge orientation and then show how to apply those results to Cuckoo hashing.

\subsection{Edge Orientation with High-Probability Worst-Case Guarantees}

We begin by considering the problem of incremental edge orientation in
a forest. Let $e_1, \ldots, e_{n - 1}$ be a sequence of edges between
vertices in $V=\{v_1, \ldots, v_n\}$ such that the edges form a tree
on the vertices. As the edges arrive online, they always form a forest
on the vertices. Each edge can be thought of as combining two trees in
the forest into one. The goal is to maintain an orientation of the
edges so that no vertex has out-degree more than three.

In Appendix \ref{sec:alg}, we present a simple Monte-Carlo randomized
algorithm, called \defn{the Dancing-Walk Algorithm}\footnote{The name
  ``Dancing-Walk'' refers to the fact that the algorithm selects a
  chain of edges to flip by performing a random walk, but the walk
  sometimes ``dances around'' rather than greedily stopping at the
  earliest available point.}, which flips at most $O(\log \log n)$
edges per edge insertion. The algorithm has worst-case operation time
$O(\log n \log \log n)$, and can be modified to take constant
amortized time per edge insertion. In this section, we give an
overview of the algorithm and its analysis.

\paragraph{Augmenting paths}
Whenever a new edge $e_t = (v_1, v_2)$ is inserted, the algorithm
first selects a \defn{source vertex} $s_t \in \{v_1, v_2\}$. The
Dancing-Walk Algorithm always selects the source vertex $s_t$ to be in
the smaller of the two (undirected) trees that are being connected by
the edge $e_t$. As a rule, the algorithm will only flip edges within
that smaller tree, and never within the larger tree; as we shall see
later, this gives the algorithm certain natural combinatorial
properties that prevent an adversary from significantly manipulating
the algorithm's behavior.

If $s_t$'s out-degree is $1$ or smaller, then the algorithm simply
inserts edge $e$ to face out of $s_t$. 
Otherwise the algorithm  selects edges to reorient in order to 
decrement $s_t$'s out-degree---after reorienting these edges, the
algorithm will then insert edge $e$ facing out of $s_t$ as before. 

In order to decrement $s_t$'s out-degree, the algorithm uses a form of
path augmentation. The algorithm finds a directed path $P_t$ of edges
from the source vertex $s_t$ to some \defn{destination vertex}
$d$ whose out-degree is smaller than $3$. The algorithm then flips every
edge in the path $P_t$, which has the effect of decrementing the
out-degree of $s_t$ and incrementing the out-degree of $d_t$.

\paragraph{The challenge: hotspot clusters of dead vertices}
A natural approach to constructing the augmenting path $P_t$ is to
simply either (a) perform a breadth-first-search to find the shortest
path to a vertex with out-degree less than $3$, or (b) perform a
random walk down out-facing edges in search of a vertex with
out-degree less than $3$.

The problem with both of these approaches is that they do nothing to
mitigate hotspots of \defn{dead vertices} (i.e., vertices with the
maximum allowable out-degree of $3$). Dead vertices are problematic
because they cannot serve as the destination in an augmenting path. If
all of the vertices near the source $s_t$ are dead (i.e., $s_t$ is in
a \defn{hotspot cluster}), then the algorithm will be forced to incur
a large number of edge-flips on a single edge-insertion.

In order to avoid the formation of dead-vertex hotspots, the algorithm
must be careful to leave vertices that are alive ``sprinkled'' around
the graph at all times. Our algorithm forces the augmenting path $P_t$
to sometimes \emph{skip over} an alive vertex for the sake of
maintaining a healthy structure within the graph. As a rule, the
algorithm is only willing to kill a vertex $v$ if $v$ is at the end of
a reasonably long random walk, in which case the vertex $v$ being
killed is sufficiently random that it can be shown not to contribute
substantially to the creation of hotspots.

\paragraph{Constructing the augmenting path}
In order to construct $P_t$, the algorithm performs a random walk
beginning at the source vertex $s_t$, and stepping along a random
outgoing edge in each step of the walk\footnote{One small but
  important technicality is that if a vertex has out-degree $3$, then
  the random walk only chooses from the first \emph{two} of the
  outgoing edges. Since the random walk terminates when it sees any
  vertices with out-degree less than 2 (we will discuss this more
  shortly), it follows that every step in the random walk chooses
  between \emph{exactly} two edges to travel down. This is important
  so that every path that the random walk could take has equal
  probability of occurring.}.

If the random walk ever encounters a vertex with out-degree less than
$2$, then that vertex is selected as the destination
vertex. Otherwise, if all of the vertices encountered have out-degrees
$2$ and $3$, then the walk continues for a total of $c \log \log n$
steps. At this point, the vertex $w$ at which the random walk resides
is asked to \defn{volunteer} as the destination vertex.

If the volunteer vertex $w$ has out-degree less than $3$ (i.e., $w$ is
still alive), then it can be used as the destination vertex for
$P_t$. Otherwise, the random walk is considered a \defn{failure} and
is restarted from scratch. If $\Theta(\log n)$ random walks in a row
fail, then the algorithm also fails.

Note that the augmenting path $P_t$ may go through many vertices with
out-degrees $2$. The only such vertex that $P_t$ will consider as a
possible destination vertex, however, is the $(c \log \log n)$-th
vertex $w$. This ensures that the algorithm avoids killing vertices in
any highly predictable fashion -- the only way that the algorithm can
kill a vertex is if that vertex is the consequence of a relatively
long random process. 

\paragraph{Analyzing candidate volunteers}
For the $t$-th edge insertion $e_t$, let $D_t$ denote the
set of \defn{candidate volunteer} vertices $w$ that can be reached
from $s_t$ by a walk consisting of exactly $c \log \log n$ steps. To simplify the discussion
for now, we treat $D_t$ as having size at least
$2^{c \log \log n}$---that is, we ignore the possibility of a random
walk hitting vertices with out-degree $1$ or $0$. Such vertices can easily be incorporated into the analysis after the fact, since they only help the random walk terminate.

To analyze the algorithm we wish to show that, with high probability
in $n$, at least a constant fraction of the vertices in
$D_t$ have never yet volunteered. This, in turn, ensures
that each random walk has a constant probability of succeeding.

\paragraph{Two key properties}
In order to analyze the fraction of the candidate-volunteer set $D_t$ that has not yet
volunteered, we use two key properties of the algorithm:

\begin{itemize}
\item \textbf{The Sparsity Property: } During the $t$-th edge insertion, each
  element in $v \in D_t$ has probability at most
  $O(1 / \log^{c - 1} n)$ of being selected to volunteer, 
  because at most $O(\log n)$ random walks are performed, and each has
  probability at most $1 / \log^c n$ of volunteering $v$.

\item \textbf{The Load Balancing Property: }Each vertex $v$ in the graph is
  contained in at most $\log n$ candidate-volunteer sets $D_t$, because, whenever a new edge $e_t$ combines two
  trees, the algorithm performs random walks only in the smaller of the two trees. It follows that a vertex $v$ can only be contained in
  $D_t$ if the size of the (undirected) tree containing $v$
  at least doubles during the $t$-th edge insertion. 
\end{itemize}

These properties imply that each vertex $v$ in the graph has probability at most
$O(1 / \log^{c - 2} n)$ of \emph{ever} volunteering. The property of
volunteering is not independent between vertices. Nonetheless, by a simple inspection of the moment generating function for the number of volunteering vertices, one can
still prove a Chernoff-style bound on them. In particular, for any fixed set of $k$ vertices,
the probability that more than half of them volunteer is\footnote{The value $c$ is a constant in that it is a parameter of the algorithm that is independent of $n$. We place $c$ within Big-O notation here in order to keep track of its impact.}
\begin{equation}
  \frac{1}{\log^{\Omega(ck)} n}.
  \label{eq:rough_chernoff}
\end{equation}
We will be setting $k$ to be $|D_t| = \log^c n$, meaning that \eqref{eq:rough_chernoff} evaluates to
\begin{equation}
  \frac{1}{\log^{\Omega(c \log^c n)} n} \ll \frac{1}{\poly (n)}.
  \label{eq:rough_chernoff2}
\end{equation}

\paragraph{A problem:  adversarial candidate sets}
If $D_t$ were a fixed set of vertices (i.e., a function only of the edge-insertion sequence $e_1, \ldots, e_{n - 1}$), then the analysis would be complete by \eqref{eq:rough_chernoff2}. The problem is that $D_t$ is \emph{not} a fixed set of vertices, that is, $D_t$ is partially a function of the algorithm's random bits and past decisions. Indeed, the decisions of which vertices have volunteered in the past affect the edge-orientations in the present, which affects the set $D_t$ of vertices that can be reached by a directed walk of length $c \log \log n$.  

In essence, $D_t$ is determined by an adaptive adversary, meaning in the worst case that $D_t$ could consist entirely of volunteered
vertices, despite the fact that the vast majority of vertices in the
graph have not volunteered. The key to completing the analysis is to show that, although $D_t$ is determined by an
adaptive adversary, the power of that adversary is severely limited by the structure of the algorithm.

\paragraph{The universe of candidate sets}
Let $\mathcal{U}_t$ denote the universe of possible candidate sets
$D_t$. That is,
$$\mathcal{U}_t = \{X \subseteq V \mid \Pr[D_t = X] > 0\}.$$
In order to complete the analysis, we prove that the universe
$\mathcal{U}_t$ is actually remarkably small. In particular,
\begin{equation}
  |\mathcal{U}_t| \le \log^{O(\log^c n)} n.
  \label{eq:bounded_universe}
\end{equation}
By \eqref{eq:rough_chernoff}, the probability that there is a set $S\in \mathcal{U}_t$ such that more than half the elements in $S$ are volunteers is at most
\begin{align*}
  \frac{|\mathcal{U}_t|}{\log^{\Omega(c \log^c n)} n} & = \frac{\log^{O(\log^c n)}n }{\log^{\Omega(c \log^c n)} n}.
\end{align*}
If $c$ is a sufficiently large constant, then the denominator
dominates the numerator. With high probability, \emph{every} set in
the universe $\mathcal{U}_t$ behaves well as an option for $D_t$. This
solves the problem of $D_t$ being potentially adversarial.

\paragraph{Bounding the universe by pre-setting children}
We prove \eqref{eq:bounded_universe} by examining the potential
\defn{children} of each vertex $v$. For a vertex $v$, the
\defn{children} of $v$ are the vertices $u$ to which $v$ has an
outgoing edge. The set of children of $v$ can change over time as edges incident to $v$ are re-oriented.

The structure of the Dancing-Walk Algorithm is designed to severely limit
the set of vertices $u$ that can \emph{ever become} children of $v$.  During the
insertion of an edge $e_t$, the only vertex that can become $v$'s
child is the vertex $u$ that appears directly before $v$ on the path
from $s_t$ to $v$. Moreover, as is argued in the Load Balancing
Property, there are only $O(\log n)$ values of $t$ for which there
even \emph{exists} a path from $s_t$ to $v$ (at the time of the edge-insertion $e_t$). Thus we have the
following property:
\begin{itemize}
\item \textbf{The Preset-Children Property:} There exists a (deterministic) set of
  $O(\log n)$ vertices $C_v$ that contains all of $v$'s potential
  children. That is, no matter what random bits the algorithm uses,
  the children of $v$ will always come from the set $C_v$.
\end{itemize}

The Preset-Children Property can be used to bound the universe size
$|\mathcal{U}_t|$ in a very simple way. Recall that the nodes in
$D_t$ are the leaves of a $c \log \log n$-level search tree
$T_t$ rooted at $s_t$. The tree $T_t$ consists of
$O(\log^c n)$ nodes. By the Preset-Children Property, each node
$v \in T_t$ has only $\binom{|C_v|}{O(1)} \le \log^{O(1)} n$
options for whom its $O(1)$ children can be in $T_t$. It follows
that the total number of possibilities for $T_t$ is at most
$$\log^{O(|T_t|)} n \le \log^{O(\log^c n)} n.$$
Each possibility for $T_t$ corresponds to a possibility for the candidate-volunteer set
$D_t$ and thus to an element of the universe $\mathcal{U}_t$. This yields
the desired bound \eqref{eq:bounded_universe} on $|\mathcal{U}_t|$.

\paragraph{Analyzing the running time}
So far we have shown that, with high probability, at least half of the
elements in the candidate-volunteer set $D_t$ are eligible to volunteer as a destination
vertex. This implies that each random walk succeeds with constant
probability, and thus that the number of failed random walks during a
given edge-insertion is $O(\log n)$ with high probability. Thus, with
high probability, the algorithm succeeds on every edge-insertion, the
running time of each edge-insertion is $O(\log n \log \log n)$, and
the number of flipped edges per edge-insertion is $O(\log \log n)$.

\paragraph{The tradeoff between edges flipped and out-degree}
Appendix \ref{sec:tradeoff} explores the tradeoff between out-degree
and the maximum number of edges that are flipped per edge insertion.

We consider a modification of the Dancing-Walk Algorithm in which nodes
are permitted to have out-degree as large as $\log^{\epsilon} n + 1$
(instead of $3$) for some parameter $\epsilon$. Rather flipping the
edges in a random walk of length $c \log\log n$, the new algorithm
instead flips the edges in a random walk of length $c
\epsilon^{-1}$. The length of the random walk is parameterized so that
the number of potential volunteers $|D_t|$ is still $\log^c n$, which
allows for the algorithm to be analyzed as in the case of out-degree
$3$. The algorithm ensures that at most $O(\epsilon^{-1})$ edges are
flipped per edge-insertion, that each edge-insertion takes time
$O(\epsilon^{-1} \log n)$, and that the total time by all
edge-insertions is $O(n)$, with high probability in $n$.

\subsection{Achieving Constant Amortized Running Time} 

In Appendix \ref{sec:amortized} we modify the Dancing-Walk Algorithm to
achieve a total running time of $X = O(n)$, with high probability in
$n$. To simplify the discussion in this section, we focus here on the
simpler problem of bounding the \emph{expected} total running time
$\E[X]$.

\paragraph{Bounding the time taken by random walks}
Although each random walk is permitted to have length as large as
$\Theta(\log \log n)$, one can easily prove that a random walk through a tree of $m$ nodes expects to hit a node with out-degree less than $2$ within $O(\log m)$ steps. Recall that,
whenever an edge $e_t$ combines two (undirected) trees $T_1$ and
$T_2$, the ensuing random walks are performed in the
\emph{smaller} of $T_1$ or $T_2$. The expected contribution to the
running time $X$ is therefore, $O(\min(\log |T_1|, \log |T_2|))$. That is, even though a given edge-insertion $e_t$ could incur up to $\Theta(\log n)$ random walks each of length $\Theta(\log \log n)$ in the worst case, the expected time spent performing random walks is no more than $O(\min(\log |T_1|, \log |T_2|))$.

Let $\mathcal{T}$ denote the set of pairs $(T_1, T_2)$ that are
combined by each of the $n - 1$ edge insertions. A simple amortized
analysis shows that
\begin{equation}
  \sum_{(T_1, T_2) \in \mathcal{T}} \min\left(\log |T_1|, \log |T_2|\right)  = O(n).
  \label{eq:amortizedlogs}
 \end{equation}
Thus the time spent performing random walks is $O(n)$ in expectation.

\paragraph{The union-find bottleneck} In addition to performing random
walks, however, the algorithm must also compare $|T_1|$ and $|T_2|$ on each edge insertion.  But maintaining a union-find data
structure to store the sizes of the trees requires
$\Omega(\alpha(n, n))$ amortized time per operation
\cite{FredmanSa89}, where $\alpha(n, n)$ is the inverse Ackermann
function. 

Thus, for the algorithm described so far, the maintenance of a union-find data structure prevents an amortized constant running time per operation. We now describe how to modify the algorithm in order to remove this bottleneck.

\paragraph{Replacing size with combination rank} We modify the
Dancing-Walk Algorithm so that the algorithm no longer needs to keep
track of the size $|T|$ of each tree in the graph. Instead the
algorithm keeps track of the \defn{combination rank} $R(T)$ of each
tree $T$---whenever two trees $T_1$ and $T_2$ are combined by an edge
insertion, the new tree $T_3$ has combination rank,
  $$R(T_3) = \begin{cases}
    \max(R(T_1), R(T_2)) & \text{ if }R(T_1) \neq R(T_2) \\
    R(T_1) + 1           & \text{ if }R(T_1) = R(T_2).
  \end{cases}$$

  Define the \defn{Rank-Based Dancing-Walk Algorithm} to be the same as
  the Dancing-Walk Algorithm, except that the source vertex $s_t$ is
  selected to be in whichever of $T_1$ or $T_2$ has
  smaller combination rank (rather than smaller size).

\paragraph{The advantage of combination rank}
The advantage of combination rank is that it can be efficiently
maintained using a simple tree structure. Using this data structure,
the time to merge two trees $T_1$ and $T_2$ (running the Dancing-Walk Algorithm
with appropriately chosen source vertex)
becomes simply
$\min(R(T_1), R(T_2))$. This, in turn, can be upperbounded by
$O(\min(\log |T_1|, \log |T_2|))$. By \eqref{eq:amortizedlogs}, the
total time spent maintaining combination ranks of trees is $O(n)$.

The other important feature of combination rank is that it preserves
the properties of the algorithm that are used to analyze
correctness. Importantly, whenever a tree $T$ is used for
path augmentation by an edge-insertion $e_t$, the combination rank of
$T$ increases due to that edge insertion. One can further prove
that the combination rank never exceeds $O(\log n)$, which allows one
to derive both the Load Balancing Property and the Preset Children
Property.

\paragraph{The disadvantage: longer random walks}
The downside of using combination rank to select trees is that
\emph{random walks} can now form a running-time bottleneck.  Whereas
the expected running time of all random walks was previously bounded
by \eqref{eq:amortizedlogs}, we now claim that it is bounded by,
\begin{equation}
  \sum_{(T_1, T_2) \in \mathcal{T}}\left( \begin{cases}
  \log |T_1| & \text{ if } R(T_1) \le R(T_2) \\
  \log |T_2| & \text{ if } R(T_2) < R(T_1) 
\end{cases}\right) = O(n).
\label{eq:newrandomwalks}
\end{equation}

We now justify this claim.

The problem is that a tree $T$ can potentially have very small
combination rank (e.g., $O(1)$) but very large size (e.g.,
$\Omega(n)$). As a result, the summation \eqref{eq:amortizedlogs} 
may differ substantially from the summation \eqref{eq:newrandomwalks}.

Rather than bounding \eqref{eq:newrandomwalks} directly, we instead
examine the smaller quantity,
\begin{equation}
  \sum_{(T_1, T_2) \in \mathcal{T}} \left( \begin{cases}
  \log |T_1| - R(T_1)& \text{ if } R(T_1) \le R(T_2) \\
  \log |T_2| - R(T_2)& \text{ if } R(T_2) < R(T_1) 
\end{cases}\right) =  O(n).
\label{eq:newrandomwalks2}
\end{equation}
The difference between \eqref{eq:newrandomwalks} and \eqref{eq:newrandomwalks2} is simply
$$\sum_{(T_1, T_2) \in \mathcal{T}}  \min\left(R(T_1), R(T_2)\right) =  O(n),$$
meaning that an upper bound on \eqref{eq:newrandomwalks2} immediately
implies an upper bound on \eqref{eq:newrandomwalks}.

The key feature of \eqref{eq:newrandomwalks2}, however, is that it yields to a simple
potential-function based analysis. In particular, if we treat each vertex $v$ as initially having $\Theta(1)$ tokens, and we treat each 
tree combination $(T_1, T_2)$ as incurring a cost given by the summand in \eqref{eq:newrandomwalks2}, then one can show that every tree $T$ always has at least
$$\Omega\left( \frac{|T|}{2^{R(T)}} \right)$$
tokens, which means that the total number of tokens spent is $O(n)$. This allows us to bound \eqref{eq:newrandomwalks2} by
  $O(n)$, which then bounds \eqref{eq:newrandomwalks} by $O(n)$, and
  implies a total expected running time of $\E[X] =  O(n)$.




  \subsection{Dynamic Cuckoo Hashing: Transforming Static Guarantees into Dynamic Guarantees}
  In Appendix \ref{sec:cuckoo}, we apply our results on
  edge-orientation to the problem of maintaining a dynamic Cuckoo hash
  table. In particular, given any hash-function family $\mathcal{H}$
  that achieves a \emph{static guarantee} in a $1$-associative Cuckoo
  hash table, we show how to achieve strong \emph{dynamic guarantees}
  in an $O(1)$-associative table. We consider a wide variety of static
  guarantees, including those that use stashes
  \cite{aumuller2014explicit, kirsch2010more} or that make assumptions
  about the input \cite{mitzenmacher2008simple}. In this section, we
  give an overview of the main ideas need to achieve these results.

  \paragraph{From hash tables to graphs}
  Say that a set $X$ of records is \defn{$(h_1, h_2)$-viable} if it is
  possible to place the records $X$ into a $1$-associative $n$-bin
  Cuckoo hash table using hash functions $h_1$ and $h_2$.

  The property of being $(h_1, h_2)$-viable has a natural
  interpretation as a graph property. Define the \defn{Cuckoo graph
    $G(X, h)$} for a set of records $X$ and a for pair of hash
  functions $h = (h_1, h_2)$ to be the graph with vertices $[n]$ and
  with (undirected) edges $\{(h_1(x), h_2(x)) \mid x \in X\}$. The
  problem of configuring where records should go in the hash table
  corresponds to an edge-orientation problem in $G$. In particular,
  one can think of each record $x$ that resides in a bin $h_i(x)$ as
  representing an edge $(h_1(x), h_2(x))$ that is oriented to face out
  of vertex $h_i(x)$. A set of records $X$ is $h$-viable if and only
  if the edges in $G$ can be oriented to so that the maximum
  out-degree is $1$.

  The fact that $G(X, h)$ can be oriented with maximum out-degree $1$
  means that $G$ is a \defn{pseudoforest} --- that is, each connected
  component of $G$ is a tree with up to one extra edge. For the sake
  of simplicity here, we will make the stronger assumption that $G$
  forms a forest; this assumption can easily be removed in any of a
  number of ways, including by simply identifying and treating
  specially any extra edges.

  \paragraph{From incremental Cuckoo hashing to incremental edge-orientation}
  If we assume that the edges in the Cuckoo graph form a forest, then
  the problem of implementing insertions in a Cuckoo hash table (for
  now we ignore deletions) is \emph{exactly} the incremental
  edge-orientation problem studied in this paper. In particular, the
  problem of finding a kickout chain corresponds exactly to the
  problem of selecting a path of edges to augment. Thus we can use the
  Rank-Based Dancing-Walk Algorithm in order to achieve a dynamic
  guarantee.

  \paragraph{Supporting deletions with phased rebuilds}
  Although our results on edge-orientation support only edge
  insertions, we wish to support both insertions an deletions in our
  hash table.

  To support deletions, we first modify the data structure so that it
  is gradually rebuilt from scratch every $\epsilon n$ insert/delete
  operations, for some $\epsilon \in (0, 1)$. By doubling the size of
  each bin, we show that these rebuilds can be performed without
  interfering with queries or inducing any high-latency
  operations. The effect of the these rebuilds is that we can analyze
  the table in independent batches of $\epsilon n$ operations.

  Consider a batch of $\epsilon n$ operations, and let $X$ denote the
  set of all records that are present during \emph{any} of those
  operations. Note that $|X|$ may be as large as $(c + \epsilon)n$
  where $cn$ is the capacity of the table. By setting
  $(c + \epsilon)n$ (rather than $cn$) to be the capacity at which
  $\mathcal{H}$ offers static guarantees, we can apply the static
  guarantee for $\mathcal{H}$ to \emph{all} of the records in $X$
  simultaneously, even though the records $X$ are not necessarily ever
  logically in the table at the same time as each other. The fact that
  $X$ is $h$-viable in its entirety (including records that are
  deleted during the batch of operations!) allows for us to analyze
  deletions without any trouble.

\bibliography{all,references}
\appendix


\section{An Algorithm with High-Probability Worst-Case Guarantees}
\label{sec:alg}

This section considers the problem of incremental edge orientation in
a forest. Let $e_1, \ldots, e_{n - 1}$ be a sequence of edges between
vertices in $V=\{v_1, \ldots, v_n\}$ such that the edges form a tree
on the vertices. 

We now present the Dancing-Walk Algorithm. The Dancing-Walk Algorithm
guarantees out-degree at most $3$ for each vertex, and performs at
most $O(\log \log n)$ edge-flips per operation. Each step of the
algorithm takes time at most $O(\log n \log \log n)$ to process. The
algorithm is randomized, and can sometimes declare failure. The main
technical difficulty in analyzing the algorithm is to show that the
probability of the algorithm declaring failure is always very small.

\paragraph{The Dancing-Walk Algorithm}
At any given moment, the algorithm allows each vertex $v$ to have up
to two \defn{primary} out-going edges, and one \defn{secondary}
out-going edge. A key idea in the design of the algorithm is that,
once a vertex has two primary out-going edges, the vertex can
\defn{volunteer} to take on a secondary out-going edge in order to
ensure that a chain of edge flips remains short. But if vertices
volunteer too frequently in some part of the graph, then the supply of
potential volunteers will dwindle, which would destroy the algorithm's
performance. The key is to design the algorithm in a way so that
volunteering vertices are able to be useful but are not overused.

Consider the arrival of a new edge $e_i$. Let $v_1$ and $v_2$ be the
two vertices that $e_i$ connects, and let $T_1$ and $T_2$ be the two
trees rooted at $v_1$ and $v_2$, respectively. The algorithm first
determines which of $T_1$ or $T_2$ is smaller (for this description we
will assume $|T_1| \le |T_2|$). Note that, by maintaining a simple
union-find data structure on the nodes, the algorithm can recover the
sizes of $T_1$ and $T_2$ each in $O(\log n)$ time.

The algorithm then performs a random walk through the (primary)
directed edges of $T_1$, beginning at $v_1$. Each step of the random
walk travels to a new vertex by going down a random outgoing primary
edge from the current vertex. If the random walk encounters a vertex
$u$ with out-degree less than $2$ (note that this vertex $u$ may even
be $v_1$), then the walk terminates at that vertex. Otherwise, the
random walk continues for a total of $c \log \log n$ steps,
terminating at some vertex $u$ with out-degree either $2$ or $3$. If
the final vertex $u$ has out-degree $2$, meaning that the vertex does
not yet have a secondary out-going edge, then the vertex $u$
volunteers to take a secondary out-going edge and have its out-degree
incremented to $3$. If, on the other hand, the final vertex $u$
already has out-degree $3$, then the random walk is considered to have
\defn{failed}, and the random walk is repeatedly restarted from
scratch until it succeeds. The algorithm performs up to $d \log n$
random-walk attempts for some sufficiently large constant $d$; if all
of these fail, then the algorithm \defn{declares failure}.

Once a successful random walk is performed, all of the edges that the
random walk traveled down to get from $v_1$ to $u$ are flipped. This
decrements the degree of $v_1$ and increments the degree of $u$. The
edge $e_i$ is then oriented to be out-going from $v_1$. The result is
that every vertex in the graph except for $u$ has unchanged
out-degree, and that $u$ has its out-degree incremented by $1$.

\paragraph{Analyzing the Dancing-Walk Algorithm}
In the rest of the section, we prove the following theorem:
\begin{theorem}
  With high probability in $n$, the Dancing-Walk Algorithm can process
  all of $e_1, \ldots, e_{n - 1}$ without declaring failure. If the
  algorithm does not declare failure, then each step flips
  $O(\log \log n)$ edges and takes $O(\log n \log \log n)$
  time. Additionally, no vertex's out-degree ever exceeds $3$.
  
  \label{thm:failure_prob}
\end{theorem}

For each edge $e_t$, let $B_t$ be the binary tree in which the random
walks are performed during the operation in which $e_t$ is
inserted. In particular, for each internal node of $B_t$, its children
are the vertices reachable by primary out-going edges; all of the
leaves in $B_t$ are either at depth $c \log \log n$, or are at smaller
depth and correspond with a vertex that has out-degree one or
zero. Note that the set of nodes that make up $B_t$ is a function of
the random decisions made by the algorithm in previous steps, since
these decisions determine the orientations of edges. Call the leaves
at depth $(c \log \log n)$ in $B_t$ the \defn{potential volunteer
  leaves}. If every leaf in $B_t$ is a potential volunteer leaf, then
$B_t$ can have as many as $(\log n)^c$ such leaves.

The key to proving Theorem \ref{thm:failure_prob} is to show with high
probability in $n$, that for each step $t$, the number of potential
volunteer leaves in $B_t$ that have already volunteered in previous
steps is at most $(\log n)^c / 2$.

\begin{proposition}
  Consider a step $t \in \{1, 2, \ldots, n - 1\}$. With high probability
  in $n$, the number of potential volunteer leaves in $B_t$ that have
  already volunteered in previous steps is at most $(\log n)^c / 2$.
  \label{prop:manyvolunteers}
\end{proposition}

Assuming the high-probability outcome in Proposition
\ref{prop:manyvolunteers}, it follows that each random walk performed
during the $t$-th operation has at least a $1/2$ chance of success. In
particular, the only way that a random walk can fail is if it
terminates at a leaf of depth $c \log \log n$ and that leaf has
already volunteered in the past. With high probability in $n$,
one of the first $O(\log n)$ random-walk attempts will succeed,
preventing the algorithm from declaring failure.

The intuition behind Proposition \ref{prop:manyvolunteers} stems from
two observations:
\begin{itemize}
\item \textbf{The Load Balancing Property: }Each vertex $v$ is
  contained in at most $\log n$ trees $B_t$. This is because, whenever
  two trees $T_1$ and $T_2$ are joined by an edge $e_t$, the tree
  $B_t$ is defined to be in the smaller of $T_1$ or $T_2$. In other
  words, for each step $t$ that a vertex $v$ appears in $B_t$,
  the size of the (undirected) tree containing $v$ at least doubles.
\item \textbf{The Sparsity Property: } During a step $t$, each
  potential volunteer leaf in $B_t$ has probability at most
  $\frac{d \log n}{\log^c n}$ of being selected to volunteer.
\end{itemize}
Assuming that most steps succeed within the first few random-walk
attempts, the two observations combine to imply that most vertices $v$
are never selected to volunteer.

The key technical difficulty comes from the fact that the structure of
the tree $B_t$, as well as the set of vertices that make up the tree,
is partially a function of the random decisions made by the algorithm
in previous steps. This means that the set of vertices in tree $B_t$
can be partially determined by which vertices have or have not
volunteered so far. In this worst case, this might result in $B_t$
consisting entirely of volunteered vertices, despite the fact that the
vast majority of vertices in the graph have not volunteered yet.

How much flexibility is there in the structure of $B_t$? One
constraint on $B_t$ is that it must form a subtree of the undirected
graph $G_t = \{e_1, \ldots, e_{t - 1}\}$. This constraint alone is not
very useful. For example, if $G_t$ is a $(\log^{c + 1} n)$-ary 
tree of depth $c \log \log n$, and if each node in $G_t$ has
volunteered previously with probability $1/\log^c n$, then there is a
reasonably high probability that every internal node of $G_t$ contains
at least two children that have already volunteered. Thus there would
exist a binary subtree of $G_t$ consisting entirely of nodes that have
already volunteered.

An important property of the Dancing-Walk Algorithm is that the tree
$B_t$ cannot, in general, form an arbitrary subtree of $G_t$. Lemma
\ref{lem:bounding_subtree_options} bounds the total number of
possibilities for $B_t$:
\begin{lemma}
  For a given sequence of edge arrivals $e_1, \ldots, e_{n - 1}$, the number
  of possibilities for tree $B_t$ is at most
  $$(\log n)^{2 \log^c n}.$$
  \label{lem:bounding_subtree_options}
\end{lemma}
\begin{proof}
  We will show that, for a given node $v$ in $B_t$, there are only
  $\log n$ options for who each of $v$'s children can be in $B_t$. In
  other words, $B_t$ is a binary sub-tree of a $(\log n)$-ary tree
  with depth $c \log \log n$. Once this is shown, the lemma can be
  proven as follows. One can construct all of the possibilities for
  $B_t$ by beginning with the root node $v_1$ and iteratively by
  adding one node at a time from the top down. Whenever a node $v$ is
  added, and is at depth less than $c \log \log n$, one gets to either
  decide that the node is a leaf, or to select two children for the
  node. It follows that for each such node $v$ there are at most
  $\binom{\log n}{2} + 1 \le \log^2 n$ options for what $v$'s set of
  children looks like. Because $B_t$ can contain at most
  $\log^c n - 1$ nodes $v$ with depths less than $c \log \log n$, the
  total number of options for $B_t$ is at most
  $\left(\log^{2} n\right)^{\log^c n}$, as stated by the lemma.

  It remains to bound the number of viable children for each node $v$
  in $B_t$. To do this, we require a stronger version of the load
  balancing property. The Strong Load Balancing Property says that,
  not only is the number of trees $B_t$ that contain $v$ bounded by
  $\log n$, but the set of $\log n$ trees $B_t$ that can contain $v$ is a
  function only of the edge sequence $(e_1, \ldots, e_{n - 1})$, and not of the
  randomness in the algorithm. 
  \begin{itemize}
  \item \textbf{The Strong Load Balancing Property: }For each vertex
    $v$, there is a set $S_v \subseteq [n]$ determined by the
    edge-sequence $(e_1, \ldots, e_{n - 1})$ such that: (1) the set's
    size satisfies $|S_v| \le \log n$, and (2) every $B_t$ containing $v$
    satisfies $t \in S_v$.
  \end{itemize}
  The Strong Load Balancing Property is a consequence of the fact
  that, whenever a new edge $e_t$ combines two trees $T_1$ and $T_2$,
  the algorithm focuses only on the smaller of the two trees. It
  follows that a vertex $v$ can only be contained in tree $B_t$ if the
  size of the (undirected) tree containing $v$ at least doubles during
  the $t$-th step of the algorithm. For each vertex $v$, there can only
  be $\log n$ steps $t$ in which the tree size containing $v$ doubles,
  which implies the Strong Load Balancing Property.

  Consider a step $t$, and suppose that step $t$ orients some edge $e$
  to be facing out from some vertex $v$. Then it must be that the path
  from edge $e_t$ to vertex $v$ goes through $e$ as its final edge. In
  other words, for a given step $t$ and a given vertex $v$, there is
  only one possible edge $e$ that might be reoriented during step $t$
  to be facing out from $v$. By the Strong Load Balancing Property, it
  follows that for a given vertex $v$, there are only $\log n$
  possibilities for out-going edges $e$. This completes the proof of
  the lemma.  
\end{proof}

Now that we have a bound on the number of options for $B_t$, the next
challenge is to bound the probability that a given option for $B_t$
has an unacceptably large number of volunteered leaves.

The next lemma proves a concentration bound on the number of
volunteered vertices in a given set. Note that the event of
volunteering is not independent between vertices. For example, if two
vertices $v$ and $u$ are potential volunteer leaves during some step,
then only one of $v$ or $u$ can be selected to volunteer during that
step.

\begin{lemma}
  Fix a sequence of edge arrivals $e_1, \ldots, e_{n - 1}$, and a set
  $S$ of vertices. The probability that every vertex in $S$ volunteers
  by the end of the algorithm is at most,
  $$O\left(\frac{1}{\log^{(c - 3)|S|} n}\right).$$
  \label{lem:bounding_volunteers}
\end{lemma}
\begin{proof}


  For each step $t \in \{1, 2, \ldots, n - 1\}$, define $F_t$ to be the
  number of elements of $S$ that are potential volunteer leaves during
  step $t$. Define
  $$p_t = \frac{F_t \cdot d \log n}{\log^c n},$$
  where $d \log n$ is the number of random-walk attempts that the
  algorithm is able to perform in each step before declaring failure. By the
  Sparsity Property, the value $p_t$ is an upper bound for the
  probability that any of the elements of $S$ volunteer during step
  $t$. In other words, at the beginning of step $t$, before any
  random-walk attempts are performed, the probability that some
  element of $S$ volunteers during step $t$ is at most $p_t$.

  Note that the values of $p_1, \ldots, p_{n - 1}$ are not known at the
  beginning of the algorithm. Instead, the value of $p_t$ is partially
  a function of the random decisions made by the algorithm in steps
  $1, 2, \ldots, t - 1$. The sum $\sum_t p_t$ is deterministically
  bounded, however. In particular, since each vertex $s \in S$ can
  appear as a potential volunteer leaf in at most $\log n$ steps (by
  the Load Balancing Property), the vertex $s$ can contribute at most
  $d \log^2 n$ to the sum $\sum_t p_t$. It follows that
  $$\sum_t p_t \le \frac{|S| d \log^2 n}{\log^c n}.$$

  Let $X_t$ be the indicator random variable for the event that some
  vertex in $S$ volunteers during step $t$. Each $X_t$
  occurs with probability at most $p_t$. The events $X_t$ are not
  independent, however, since the value of $p_t$ is not known until
  the end of step $t - 1$. Nonetheless, the fact that $\sum_t p_t$ is
  bounded allows for us to prove a concentration bound on $\sum_t X_t$
  using the following claim.

  \begin{claim}
    Let $\mu \in [0, n]$, and suppose that Alice is allowed to select
    a sequence of numbers $p_1, p_2, \ldots, p_k$, $p_i \in [0, 1]$,
    such that $\sum_i p_i \le \mu$. Each time Alice selects a number
    $p_i$, she wins $1$ dollar with probability $p_i$. Alice is an
    adaptive adversary in that she can take into account the results
    of the first $i$ bets when deciding on $p_{i + 1}$. If $X$ is
    Alice's profit from the game,
  $$\Pr\Big[X > (1 + \delta) \mu\Big] \le \exp\left((\delta - \ln(1 +
    \delta) (1 + \delta))\mu\right),$$
  for all $\delta  > 0$.
    \label{clm:chernoffbound}
  \end{claim}

  The proof of Claim \ref{clm:chernoffbound} follows by inspection of
  the moment generating function for $X$, and is deferred to Appendix
  \ref{app:chernoffbound}.

  Applying Claim \ref{clm:chernoffbound} to $X = \sum_t X_t$, with
  $\delta = \frac{\log^c n}{d \log^2 n} - 1$ and
  $\mu = \frac{|S| d \log^2 n}{\log^c n}$ (so that
  $(\delta + 1) \mu = |S|$), we get that
  \begin{align*}
    \Pr[X > |S|] & \le \exp\left(|S| - |S| \ln \frac{\log^c n}{d \log^2 n}\right) \\
                 & = O\left(\exp\left(-|S| \ln \log^{c - 3} n\right)\right) \\
                 & = O\left(\log^{-(c - 3)|S|} n\right).
  \end{align*}
  
\end{proof}

Combining Lemmas \ref{lem:bounding_subtree_options} and
\ref{lem:bounding_volunteers}, we can now prove Proposition
\ref{prop:manyvolunteers}.

\begin{proof}[Proof of Proposition \ref{prop:manyvolunteers}]
  Consider a tree $B_t$. By Lemma \ref{lem:bounding_subtree_options},
  the number of options for $B_t$, depending on the behavior of the
  algorithm in steps $1, 2, \ldots, t - 1$, is at most,
  $$(\log n)^{2 \log^c n}.$$
  For a given choice of $B_t$, there are at most
  $\binom{\log^c n}{\frac{1}{2} \log^c n} \le 2^{\log^c n}$ ways to
  choose a subset $S$ consisting of $\frac{\log^c n}{2}$ of the
  potential volunteer leaves. For each such set of leaves $S$, Lemma
  \ref{lem:bounding_volunteers} bounds the probability that all of the
  leaves in $S$ have already volunteered by,
  $$O\left(\log^{-(c - 3)|S|} n\right) = O\left(\log^{-(\log^c n)(c - 3)/2} n\right).$$
  Summing this probability over all such subsets $S$ of all
  possibilities for $B_t$, the probability that $B_t$ contains
  $\frac{\log^c n}{2}$ already-volunteered leaves is at most,
  \begin{align*}
    &O\left((\log n)^{2 \log^c n} \cdot 2^{\log^c n} \cdot \log^{-(\log^c n)(c - 3)/2} n\right) \\
    & = O\left(\frac{(2\log n)^{2 \log^c n}}{\log^{(\log^c n)(c - 3)/2} n}\right). \\
  \end{align*}
  For a sufficiently large constant $c$, this is at most
  $\frac{1}{n^{\omega(1)}}$.  The proposition follows by taking a
  union bound over all $t \in \{1, 2, \ldots, n - 1\}$.
\end{proof}

We conclude the section with a proof of Theorem \ref{thm:failure_prob}
\begin{proof}[Proof of Theorem \ref{thm:failure_prob}]
  Consider a step $t$ in which the number of potential volunteer
  leaves in $B_t$ that have already volunteered is at most
  $\frac{1}{2} \log^c n$. The only way that a random walk in step $t$
  can fail is if it lasts for $c \log \log n$ steps (without hitting a
  vertex with out-degree $1$ or $0$) and it finishes at a vertex that
  has already volunteered. It follows that, out of the $\log^c n$
  possibilities for a $(c\log\log n)$-step random walk, at most half
  of them can result in failure. Since each random-walk attempt
  succeeds with probability at least $1/2$, and since the algorithm
  performs up to $d \log n$ attempts for a large constant $d$, the
  probability that the algorithm fails on step $t$ is at most
  $\frac{1}{n^d} = \frac{1}{\poly n}$.

  The above paragraph establishes that, whenever the search tree $B_t$
  contains at most $\frac{1}{2}\log^c n$ potential volunteer leaves
  that have already volunteered, then step $t$ will succeed with high
  probability in $n$. It follows by Proposition
  \ref{prop:manyvolunteers} that every step succeeds with high
  probability in $n$.

  We complete the theorem by discussing the properties of the
  algorithm in the event that it does not declare failure. Each step
  flips at most $O(\log \log n)$ edges and maintains maximum
  out-degrees of $3$. Because each step performs at most $O(\log n)$
  random-walk attempts, these attempts take time at most
  $O(\log n \log \log n)$ in each step. Additionally, a union-find
  data structure is used in order to allow for the sizes $|T_1|$ and
  $|T_2|$ of the two trees being combined to be efficiently computed
  in each step. Because the union-find data structure can be
  implemented to have worst-case operation time $O(\log n)$, the
  running time of each edge-insertion remains at most
  $O(\log n \log \log n)$.
\end{proof}

\section{Achieving Constant Amortized Running Time}
\label{sec:amortized}

Although Theorem \ref{thm:failure_prob} bounds the worst-case running
time of operations (with high probability), it does not bound the
amortized running time of the Dancing-Walk Algorithm. In this section, we
show how to modify the Dancing-Walk Algorithm so that Theorem
\ref{thm:failure_prob} continues to hold, and so that the amortized
cost of performing $n$ edge insertions is $O(n)$ with high probability
in $n$.

\paragraph{The Initial Union-Find Bottleneck} Recall that whenever
an edge $e_i$ is inserted, the Dancing-Walk Algorithm begins the
operation by determining which of the two trees $T_1$ and $T_2$ that
are being combined are smaller. In order to do this, the Dancing-Walk
Algorithm maintains a union-find data structure, storing the size of
each (undirected) tree. Maintaining such a data structure is not
viable if we wish to perform operations in constant amortized time,
however, since performing $n$ unions and $n$ finds with a union-find
data structure requires $\Omega(\alpha(n, n))$ amortized time per
operation \cite{FredmanSa89}, where $\alpha(n, n)$ is the inverse Ackermann function.

\paragraph{Replacing Size with Combination Rank} We now modify the Dancing-Walk
Algorithm so that the algorithm no longer needs to keep track of the
size $|T|$ of each tree in the graph. Instead the algorithm keeps
track of the \defn{combination rank} $R(T)$ of each tree $T$, which we
define recursively as follows:
\begin{itemize}
\item The combination rank of a tree $T$ of size $1$ is $R(T) = 0$.
\item Whenever two trees $T_1$ and $T_2$ are combined by an edge
  insertion, the new tree $T_3$ has combination rank,
  $$R(T_3) = \begin{cases}
    \max(R(T_1), R(T_2)) & \text{ if }R(T_1) \neq R(T_2) \\
    R(T_1) + 1           & \text{ if }R(T_1) = R(T_2).
  \end{cases}$$
\end{itemize}

Define the \defn{Rank-Based Dancing-Walk Algorithm} to be the same as the
Dancing-Walk Algorithm, except that whenever two trees $T_1$ and
$T_2$ are combined by an edge-insertion, the tree with smaller
combination rank (rather than the tree with smaller size) is one in
which random-walk searches are performed. As in the Dancing-Walk
Algorithm, ties can be broken arbitrarily.

\paragraph{Correctness of the Rank-Based Dancing-Walk Algorithm}
Before describing how to efficiently implement the Rank-Based
Dancing-Walk Algorithm, we first prove its correctness.

\begin{lemma}
  With high probability in $n$, the Rank-Based Dancing-Walk Algorithm can
  process all of $e_1, \ldots, e_{n - 1}$ without declaring
  failure. If the algorithm does not declare failure, then each step
  flips $O(\log \log n)$ edges and takes $O(\log n \log \log n)$
  time. Additionally, no vertex's out-degree ever exceeds $3$.
  \label{lem:rankcorrectness}
\end{lemma}
\begin{proof}
  In order for the proof to follow just as in Theorem
  \ref{thm:failure_prob}, it suffices to show that the Strong Load
  Balancing Property holds for the Rank-Based Dancing-Walk Algorithm.

  Note that the rank $R(T)$ of a tree $T$ is determined entirely by
  the edge-sequence $(e_1, \ldots, e_{n - 1})$. It therefore suffices
  to show that each vertex $v$ appears in at most $\log n$ search
  trees.

  Whenever a vertex $v$ appears in a search tree for some step $t$,
  the combination rank of the tree containing $v$ increases during
  that step $t$. Thus it suffices to bound the maximum combination
  rank by $\log n$.

  To bound the maximum combination rank, we observe as an invariant
  that $R(T)$ never exceeds $\log |T|$ for any tree $T$. To prove the
  invariant, consider two trees $T_1$ and $T_2$ that are combined by
  an edge-insertion. If $R(T_1) \neq R(T_2)$, then the new tree $T_3$
  will have rank
  $R(T_3) = \max(R(T_1), R(T_2)) \le \max(\log |T_1|, \log |T_2|) \le
  \log |T_3|$, as desired. On the other hand, if $R(T_1) = R(T_2)$,
  then the new tree $T_3$ will have rank
  $R(T_3) = R(T_1) + 1 \le \log |T_1| + 1$. It we set $T_1$ to be 
  the smaller of the two trees $T_1$ and $T_2$, then it follows that
  $R(T_3) \le \log |T_3|$. This completes the proof of the invariant,
  which bounds the maximum combination rank by $\log n$, thereby
  establishing the Strong Load Balancing Property, as desired.
\end{proof}

\paragraph{Efficient Computation of Combination Rank}
As the edges $e_1, \ldots, e_{n - 1}$ arrive, the Rank-Based Dancing-Walk
Algorithm maintains a \defn{combination rank-maintenance data structure}, which
begins with $n$ vertices (i.e., $n$ rank-$0$ trees), and supports a
single operation:
\begin{itemize}
\item \textbf{Combine($v_1, v_2$), where $v_1, v_2$ are vertices.} If
  $T_1$ and $T_2$ are the connected components containing $v_1$ and
  $v_2$, respectively, then this function determines which of $R(T_1)$
  or $R(T_2)$ is smaller (breaking ties arbitrarily). If
  $T_1 \neq T_2$ then $T_1$ and $T_2$ are then combined to a single
  component, and otherwise the fact that $T_1 = T_2$ is reported to
  the user.
\end{itemize}
Each time that an edge $e_i = (v_1, v_2)$ is inserted, the function
Combine($v_1, v_2$) is invoked by the Rank-Based Dancing-Walk Algorithm
in order to determine which tree to perform random-walk searches in.

Lemma \ref{lem:combrank} gives a simple data structure for
efficiently implementing combination rank-maintenance.
\begin{lemma}
  The combination rank-maintenance data structure can be implemented in space
  $O(n)$ so that Combine($v_1, v_2$) takes time
  $O(\min(R(T_1), R(T_2)))$ and incurs at most $O(1)$ writes.
  \label{lem:combrank}
\end{lemma}
\begin{proof}
  The combination rank-maintenance data structure stores all of the vertices in
  each connected component $T$ in what we call a \defn{rank tree}. For
  a given connected component $T$, all of the vertices in $T$ are
  leaves in $T$'s rank tree, and all of the leaves appear the same
  depth $R(T)$. This means that, given a vertex $v \in T$, one can
  compute $R(T)$ in time $O(R(T))$ by following a leaf-to-root path in
  the rank tree.

  In order to combine two components $T_1$ and $T_2$ such that
  $R(T_1) < R(T_2)$, we simply add a pointer from the root of the
  rank-tree for $T_1$ to any node in $T_2$ at height $R(T_1) + 1$
  above the leaves. In order to combine two components $T_1$ and $T_2$
  such that $R(T_1) = R(T_2)$, we simply add a new root node $r$ and
  add pointers from the roots of the rank trees for $T_1$ and $T_2$ to
  $r$. In both cases, the rank tree for $T_1 \cup T_2$ can be computed
  in time $\min(R(T_1), R(T_2))$ from the rank trees for $T_1$ and
  $T_2$. It follows that, given two vertices $v_1$ and $v_2$ appearing
  in connected components $T_1$ and $T_2$, we can perform
  Combine($v_1, v_2$) in time $\min(R(T_1), R(T_2))$.

  Each call to Combine($v_1, v_2$) adds at most $O(1)$ new pointers to
  the data structure, requiring at most $O(1)$ writes. Since the
  leaves of the rank trees are the $n$ vertices in the graph, the sum
  of the sizes of the rank trees is $O(n)$. Thus the combination rank-maintenance
  data structure takes space $O(n)$, as desired.
\end{proof}

\paragraph{An Amortized Running-Time Analysis}
In the rest of this section, we give an amortized analysis of the
Rank-Based Dancing-Walk Algorithm. The first step in the analysis is to
bound the total time needed for all of the operations in the
combination rank-maintenance data structure. 

Let $\mathcal{T}$ be the set of pairs $(T_1, T_2)$ such that for some
step $t$, trees $T_1$ and $T_2$ are connected components in the graph
$(V, \{e_1, \ldots, e_{t - 1}\})$ and are combined by edge $e_i$ into
a single tree. The order of each pair (i.e., $(T_1, T_2)$ vs
$(T_2, T_1)$) is selected so that $|T_1| \le |T_2|$, with ties broken
arbitrarily.

Each combination $(T_1, T_2)$ results in a rank-maintenance operation
that costs $O(\min(R(T_1), R(T_2))$. Lemma \ref{lem:sum_ranks} shows
that the sum of these costs is $O(n)$.

\begin{lemma}
  $$\sum_{(T_1, T_2) \in \mathcal{T}} \min(R(T_1), R(T_2)) =  O(n).$$
  \label{lem:sum_ranks}
\end{lemma}
\begin{proof}
  Recall from the proof of Lemma \ref{lem:rankcorrectness} that
  $R(T_1) \le \log |T_1|$ and $R(T_2) \le \log |T_2|$. It follows that,
  \begin{align*}
     \sum_{(T_1, T_2) \in \mathcal{T}} \min(R(T_1), R(T_2)) 
    & \le \sum_{(T_1, T_2) \in \mathcal{T}} \min(\log |T_1|, \log |T_2|) \\
    & = \sum_{(T_1, T_2) \in \mathcal{T}} \log |T_1|. \\
  \end{align*}

  Rearranging the above sum to be from the perspective of vertices gives,
  \begin{equation}
     \sum_{v \in V} \phantom{f}\sum_{(T_1, T_2) \in \mathcal{T} \text{ s.t. } v \in T_1} \frac{\log |T_1|}{|T_1|}. \\
    \label{eq:vertex_log_sum}
  \end{equation}
  Each time that vertex $v$ appears in $T_1$ for some pair
  $(T_1, T_2) \in \mathcal{T}$, the combined tree $T_1 \cup T_2$ has
  size at least twice as large as $|T_1|$. It follows that for each
  power of two $2^k$, $v$ appears in at most one pair $(T_1, T_2)$
  where $|T_1| \in [2^{k - 1}, 2^k)$. Thus \eqref{eq:vertex_log_sum}
  is at most,
  \begin{align*}
     \sum_{v \in V} \sum_{k = 1}^{\lceil \log n \rceil} \frac{k }{2^{k-1}} 
    & =  \sum_{v \in V} O(1) \\ & = O(n).
    \label{eq:vertex_log_sum}
  \end{align*}
\end{proof}

Next we bound the total time required for all of the the random-walk
attempts to be performed by the algorithm. Since every edge-insertion
results in at least one random-walk attempt, it does not suffice to
simply bound the time for each random-walk attempt by
$O(\log \log n)$.

Consider the tree $B_t$ in which a random-walk attempt is performed.
Intuitively, if the tree $B_t$ is very small, then the first random-walk
attempt should terminate in $o(\log \log n)$ steps, having arrived at
a leaf of the tree. Lemma \ref{lem:bound_random_walk_length} captures
this formally, bounding the length of the random-walk attempt by a
geometric random variable with expected value $O(\log |B_t|)$.

\begin{lemma}
  Consider a random-walk attempt performed in tree $B_t$. For any
  $k \in \mathbb{N}$, the probability that the random walk lasts for
  more than $4 k \log |B_t|$ steps is at most $\frac{1}{2^k}$.
  \label{lem:bound_random_walk_length}
\end{lemma}
\begin{proof}
  Define $s_1, s_2, \ldots$ so that if the random walk is at vertex
  $v$ at the beginning of its $i$-th step, then $s_i$ is the size of
  the subtree rooted at $v$ in $B_t$. Each step in the random walk has
  at least a $\frac{1}{2}$ probability of reducing the size of the
  subtree in which it resides by at least a factor of two. In other
  words, each $s_i$ has at least a $\frac{1}{2}$ probability of
  satisfying $s_i \le \frac{1}{2}s_{i - 1}$. After $4 k \log |B_t|$
  steps, the expected number of steps $s_i$ for which
  $s_i \le \frac{1}{2} s_{i - 1}$ is at least $2 k \log |B_t|$. In order for
  the random walk to have not terminated, the number of steps $s_i$
  for which $s_i \le \frac{1}{2} s_{i - 1}$ must be at most
$\log|B_t|$. By a (very loosely applied) Chernoff bound, the
probability that a sum $R$ of independent indicator random variables
with total mean $\E[R] = 2k \log |B_t|$ has value $R \le \log|B_t|$,
is at most $\frac{1}{2^k}$.
\end{proof}

By Lemma \ref{lem:bound_random_walk_length}, the random-walk attempts
for each edge insertion $e_t$ will take time at most $O(\log |B_t|)$
in expectation. This, in turn, is at most $O(\log |T|)$, where $T$ is
the connected component of $(V, \{e_1, \ldots, e_{t - 1}\})$
containing $B_t$. 

In order to bound the total time required by the random walks, we wish
to prove that,
\begin{equation}
  \left(\sum_{(T_1, T_2) \in \mathcal{T}}
\begin{cases}
  \log |T_1| & \text{ if } R(T_1) \le R(T_2) \\
  \log |T_2| & \text{ if } R(T_2) < R(T_1) 
\end{cases}\right) =  O(n).
\label{eq:sum_log_sizes}
\end{equation}
Recall from the proof of Lemma \ref{lem:rankcorrectness} that
$R(T) \le \log |T|$ for each tree $T$. Thus \eqref{eq:sum_log_sizes}
is a stronger inequality than the one proven in Lemma
\ref{lem:sum_ranks}.

In some cases, the combination rank of the tree $T$ could be
significantly smaller than $\log |T|$. For example, if tree $T$ has
combination rank $1$, and $\Omega(n)$ trees with combination ranks $0$
are combined with $T$, then $T$ could be of size $\Omega(n)$ while
still having combination rank only $1$. One consequence of this is
that, for a pair $(T_1, T_2) \in \mathcal{T}$, it may be that
$R(T_2) \ll R(T_1)$ but that $\log |T_2| \gg \log |T_1|$, meaning that
the algorithm selects the tree $T_2$ to perform random-walk attempts
in, even though $T_1$ would have been a better choice.

Lemma \ref{lem:sum_log_sizes} uses a potential-function argument to
prove \eqref{eq:sum_log_sizes}.
\begin{lemma}
  $$\left(\sum_{(T_1, T_2) \in \mathcal{T}}
  \begin{cases}
    \log |T_1| & \text{ if } R(T_1) \le R(T_2) \\
    \log |T_2| & \text{ if } R(T_2) < R(T_1) 
  \end{cases}\right) = O(n).$$
  \label{lem:sum_log_sizes}
\end{lemma}
\begin{proof}
  By Lemma \ref{lem:sum_ranks}, it suffices to show that,
  \begin{equation}
    \left(\sum_{(T_1, T_2) \in \mathcal{T}}
  \begin{cases}
    \log |T_1| - R(T_1) & \text{ if } R(T_1) \le R(T_2) \\
    \log |T_2| - R(T_2) & \text{ if } R(T_2) < R(T_1) 
  \end{cases}\right) =  O(n).
  \label{eq:sum_diff}
\end{equation}

We prove \eqref{eq:sum_diff} by an amortization argument. We begin by
assigning some large positive constant number $\rho$ of tokens to each
vertex $v$. Whenever two trees $T_1$ and $T_2$ are combined, the new
tree $T_3 = T_1 \cup T_2$ is given the tokens from each of $T_1$ and
$T_2$, and then pays
$$\begin{cases}
    \log |T_1| - R(T_1) & \text{ if } R(T_1) \le R(T_2) \\
    \log |T_2| - R(T_2) & \text{ if } R(T_2) < R(T_1). 
  \end{cases}$$ tokens to the algorithm.

  In order to prove \eqref{eq:sum_diff}, we wish to prove that the
  final tree consisting of all edges $\{e_1, \ldots, e_{n - 1}\}$ has
  a non-negative number of tokens. This means that the total token
  expenditure due to all combinations is at most $\rho n = O(n)$.

  We prove as an invariant that whenever a new tree $T$ is created, it
  has at least $\rho \frac{|T|}{2^{R(T)}}$ tokens. As a base case,
  this is true for trees $T$ consisting of a singleton node $v$, since
  each such tree initially has $\rho$ tokens.

  Consider a pair of trees $T_1, T_2$ that are combined by some
  edge-insertion $e_t$, and let $T_3$ be the tree that combines
  them. Let $R_i = R(T_i)$ and $S_i = |T_i|$ for $i \in \{1, 2,
  3\}$. We are given as an inductive hypothesis that $T_1$ has at
  least $\rho \frac{S_1}{2^{R_1}}$ tokens and that $T_2$ has at least
  $\rho \frac{S_2}{2^{R_2}}$ tokens. We wish to show that $T_3$ has at
  least $\rho \frac{S_3}{2^{R_3}}$ tokens.

  We begin by considering the case where $R_1 \neq R_2$, and we assume
  without loss of generality that $R_1 < R_2$. This means that
  $R_3 = R_2$ and $S_3 = S_1 + S_2$. By the inductive hypotheses, the
  number of tokens that $T_3$ has is at least
  \begin{align*}
    \rho\frac{S_1}{2^{R_1}} + \rho\frac{S_2}{2^{R_2}} - \log \left(S_1 / 2^{R_1}\right) 
    & = \rho 2^{R_2 - R_1}\frac{S_1}{2^{R_2}} + \rho \frac{S_2}{2^{R_2}} - \log \left(S_1 / 2^{R_1}\right) \\
    & = \rho \frac{S_1 + S_2}{2^{R_2}} + \rho \left(2^{R_2 - R_1} - 1\right)\frac{S_1}{2^{R_2}} - \log \left(S_1 / 2^{R_1}\right) \\
    & = \rho \frac{S_3}{2^{R_3}} + \rho \left(2^{R_2 - R_1} - 1\right)\frac{S_1}{2^{R_2}} - \log \left(S_1 / 2^{R_1}\right). \\
  \end{align*}
  In order to complete the argument, we wish to show that
  \begin{equation}
    \rho \left(2^{R_2 - R_1} - 1\right)\frac{S_1}{2^{R_2}} \ge \log \left(S_1 / 2^{R_1}\right).
    \label{eq:Runequalinequality}
  \end{equation}
  To prove this, note that
  \begin{align*}
    \rho \left(2^{R_2 - R_1} - 1\right)\frac{S_1}{2^{R_2}} 
    & = \rho \left(1 - 2^{R_1 - R_2}\right)\frac{S_1}{2^{R_1}} \\
    & \ge \frac{1}{2} \rho \frac{S_1}{2^{R_1}}. \\
  \end{align*}
  Assuming $\rho$ is a sufficiently large constant,
  $\frac{1}{2} \rho x \ge \log x$ for any $x \ge 1$. Thus
  \eqref{eq:Runequalinequality} holds, implying that $T_3$ has at
  least $\rho \frac{S_3}{2^{R_3}}$ tokens, as desired. 

  Next we consider the case where $R_1 = R_2 = R$ for some $R$. Then
  the new tree $T_3$ has rank $R_3 = R + 1$ and size $S_3 = S_1 + S_2$. The number of
  tokens that $T_3$ has is at least,
  \begin{align*}
     \rho \frac{S_1}{2^{R}} + \rho \frac{S_2}{2^{R}}  - \log (S_1 / 2^{R}) 
    & = \rho \frac{S_1 + S_2}{2^{R}} - \log (S_1 / 2^{R}) \\
    & = \rho \frac{S_3}{2^{R_3}} + \rho \frac{S_1 + S_2}{2^{R + 1}} - \log (S_1 / 2^{R}).
  \end{align*}
  As long as $\rho$ is a sufficiently large constant, then
  $\rho \frac{S_1 + S_2}{2^{R + 1}} \ge \log (S_1 / 2^{R})$. Thus
  $S_3$ has at least $\rho \frac{S_3}{2^{R_3}}$ tokens, as desired.

  This completes the proof that every tree $T$ has at least
  $\rho \frac{|T|}{2^{R(T)}}$ tokens. Since the system begins with
  $O(n)$ tokens, and ends with a non-negative number of tokens, the
  total number of tokens spent must be $O(n)$. Thus the lemma is proven.
\end{proof}

Combining the preceding lemmas, we can now analyze the total time for
algorithm to perform all of the edge insertions
$e_1, \ldots, e_{n - 1}$.

\begin{theorem}
  To perform $n - 1$ edge insertions, the total time required by the
  Rank-Based Dancing-Walk Algorithm is at most $O(n)$ with high
  probability in $n$.
  \label{thm:amortized}
\end{theorem}
\begin{proof}
  By Lemma \ref{lem:combrank}, the combination rank-maintenance data structure
  takes time $O(\min(R(T_1), R(T_2)))$ to combine two tree $T_1$ and
  $T_2$. By Lemma \ref{lem:sum_ranks}, the total time taken by the
  data structure across all operations is $O(n)$.

  Let $q_t$ be the sum of the lengths of the random-walk attempts for
  each edge-insertion $t$.  Lemma \ref{lem:bound_random_walk_length}
  bounds $q_t$ by a geometric random variable with mean $O(\log |T|)$,
  where $T$ is the connected component of
  $(V, \{e_1, \ldots, e_{t - 1}\})$ in which the random-walk attempt
  is performed. It follows by Lemma \ref{lem:sum_log_sizes} that
  $$\E[\sum_t q_t] =  O(n).$$
  Moreover, regardless of the values of
  $q_1, \ldots, q_{t - 1}, q_{t + 1}, \ldots, q_{n - 1}$, the value of
  $q_t$ is guaranteed by Lemma \ref{lem:bound_random_walk_length} to
  be bounded above by a geometric random variable with expected value
  $O(\log|T|)$, and $q_t$ is also guaranteed to be deterministically
  at most $O(\log n \log \log n)$. Applying Hoeffding's Inequality, it
  follows that the probability of $\sum_t q_t$ deviating from its mean
  by more than $\Omega(n)$ is at most
  $\exp\left(-\tilde{\Omega}(n)\right)$, completing the proof of the
  theorem.

\end{proof}

\begin{remark}
  We remark that, in order to simplify the amortized analysis above,
  one can instead make the Rank-Based Dancing-Walk Algorithm slightly
  more complicated, and allow for a larger maximum out-degree, in order that the time to combine two trees $T_1$ and $T_2$ is at most $O(\min(\log |T_1|, \log |T_2|))$ in expectation.

  A first attempt at doing this might be to perform two random walks
  in parallel in each of $T_1$ and $T_2$. This would ensure that the
  expected time for the edge-insertion would be at most
  $O(\min(\log |T_1|, \log |T_2|))$. This modification to the
  algorithm breaks the Load Balancing Property, however, eliminating
  the proof of algorithm correctness.

  In order to rescue algorithm correctness, one can further modify the
  algorithm to allow for maximum out-degree $5$. One can then maintain
  two edge-orientations, one of maximum out-degree $3$ (that is
  maintained by the Rank-Based Dancing-Walk Algorithm), and one of
  maximum out-degree $2$ (which is maintained by performing random
  walks in both directions until a vertex with out-degree $1$ or
  smaller is found). Whenever an edge is inserted, it is inserted into
  both edge-orientations in parallel, and whichever edge-orientation
  completes first is the one that keeps the edge. This allows for an
  expected running time of at most $O(\min(\log |T_1|, \log |T_2|))$
  for an operation that combines two trees $T_1$ and $T_2$, while also
  maintaining the Strong Load Balancing Property in the
  edge-orientation with out-degree $3$.
\end{remark}

\section{A Tradeoff Curve Between Out-Degree and Number of Edges Flipped}\label{sec:tradeoff}

In this section, we consider a variant of the Rank-Based Dancing-Walk
Algorithm in which the maximum out-degree is permitted to be a larger
value $k + 1$. In particular, each vertex is now permitted up to $k$
primary out-going edges, and $1$ secondary out-going edge. Each step
of each random-walk search now selects one of $k$ edges to travel
down. If a random-walk search reaches a vertex with fewer than $k$
primary out-going edges, then the random walk succeeds and can stop at
that vertex. Otherwise, random walks last for $c\log_k \log n$ steps,
and succeed if they terminate at a vertex that has not yet volunteered
(i.e., a vertex that does not yet have a secondary out-going edge).

Note that the search tree $B_t$ is now a $k$-ary tree. The number of
potential volunteer leaves, however, remains as it was before, since
$$k^{c \log_k \log n} = \log^c n.$$

In order to analyze the new algorithm, the key observation that one
must make is that Lemma \ref{lem:bounding_subtree_options}, which
bounds the number of options for $B_t$, continues to hold exactly as
stated. In particular, $B_t$ is now a $k$-ary subtree of a
$(\log n)$-ary tree with depth $c \log_k \log n$. This means that each
node in $B_t$ that has depth less than $c \log_k \log n$ has up to
$\binom{\log n}{k} + 1 \le \log^k n$ options for what its set of
children can look like (with the set either being empty or being of
size $k$). Since the number of nodes in $B_t$ with depth less than
$c \log_k \log n$ is at most $2\frac{\log^c n}{k}$, the total number of
options for $B_t$ is at most,
$$\left(\log^k n\right)^{2\frac{\log^c n}{k}} \le \log^{2\log^c n} n,$$
which is precisely the bound shown by Lemma
\ref{lem:bounding_subtree_options} in the special case of $k = 2$.

Besides the proof of Lemma \ref{lem:bounding_subtree_options}, the
analysis of the Rank-Based Dancing-Walk Algorithm generalizes without
modification to apply to the new algorithm. Thus we arrive at the
following theorem:

\begin{theorem}
  Consider the Rank-Based Dancing-Walk Algorithm with maximum out-degree
  $k + 1$.  With high probability in $n$, the algorithm can process
  all of $e_1, \ldots, e_{n - 1}$ without declaring failure. If the
  algorithm does not declare failure, then each step flips
  $O(\log_k \log n)$ edges and takes $O(\log n \log_k \log n)$
  time. Additionally, no vertex's out-degree ever exceeds $k+1$.

  Additionally, the total running time of the algorithm to perform all
  edge insertions is at most $O(n)$, with high probability in $n$.
  \label{thm:tradeoff}
\end{theorem}

One interesting case of Theorem \ref{thm:tradeoff} is when
$k = \log^{1 / q} n$ for some value $q$. In this case, the algorithm
achieves maximum out-degree $\log^{1/q} n + 1$ while flipping only $O(q)$
edges per edge-insertion.

\section{Dynamic Cuckoo Hashing: Transforming Static Guarantees into Dynamic Guarantees}
\label{sec:cuckoo}

In this section we present the \defn{Dancing-Kickout Algorithm}
for maintaining a Cuckoo hash table. For any family of hash functions
$\mathcal{H}$ that provides a $1$-associative static guarantee, the
Dancing-Kickout Algorithm offers a $O(1)$-associative dynamic
guarantee using the same hash-function family $\mathcal{H}$.

\paragraph{Allowing for a stash} We will state our results so that
they also apply to Cuckoo hashing with a stash
\cite{aumuller2014explicit, kirsch2010more}. A Cuckoo hash table with
a \defn{stash} of size $s$ is permitted to store $s$ elements
\emph{outside} of the table in a separate list. Having a small stash
has been shown by past work to significantly simplify the problem of
achieving high-probability static guarantees
\cite{aumuller2014explicit} -- our results can be used to make these
guarantees dynamic.

\paragraph{What static guarantees promise: viability}
Let $h = (h_1, h_2)$ be a pair of hash functions mapping records to
$[n]$. A set $X$ of records is \defn{$h$-viable} if it is possible
to place the records $X$ into a $1$-associative $n$-bin Cuckoo hash
table using hash functions $h_1$ and $h_2$.

Even if a set of records $X$ is not $h$-viable, it may be that there
is a set of $s$ elements $Y$ for which $X \setminus Y$ is
$h$-viable. In this case, we say $X$ is \defn{$h$-viable with a stash
  of size $s$}.

\paragraph{Past static guarantees}
Past static guarantees \cite{mitzenmacher2008simple,
  patracscu2012power,aamand2018power,
  aumuller2014explicit,dietzfelbinger2003almost,aumuller2016simple,pagh2001cuckoo}
for a hash family $\mathcal{H}$, have taken the following form, where
$c \in (0, 1), p(n) \in \poly(n), s \in O(1)$ are parameters:
Every set of records $X$ of size $c n$ has probability at least
$1 - 1/p(n)$ of being $h$-viable with a stash of size $s$, where
$h = (h_1, h_2)$ is drawn from $\mathcal{H}$. In addition to
considering guarantees of this type, a fruitful line of work
\cite{mitzenmacher2008simple} has also placed additional
restrictions on the set $X$ of records (namely, that $X$ exhibits high
entropy). In this section, we will state our results in such a way so
that they are applicable to all of the past variants of static
guarantees that we are aware of.

\paragraph{Viability as a graph property}
 Define the \defn{Cuckoo graph
    $G(X, h)$} for a set of records $X$ and  for a pair of hash
  functions $h = (h_1, h_2)$ to be the graph with vertices $[n]$ and
  with (undirected) edges $\{(h_1(x), h_2(x)) \mid x \in X\}$. The
  problem of configuring where records should go in the hash table
  corresponds to an edge-orientation problem in $G$. In particular,
  one can think of each record $x$ that resides in a bin $h_i(x)$ as
  representing an edge $(h_1(x), h_2(x))$ that is oriented to face out
  of vertex $h_i(x)$. A set of records $X$ is $h$-viable if and only
  if the edges in $G(X, h)$ can be oriented to so that the maximum
  out-degree is $1$.

Similarly, a set of records $X$ is $h$-viable with a stash of size $s$
if and only if there are $s$ (or fewer) edges that can be removed from
the Cuckoo graph $G(X, h)$ so that the new graph $G'$ can be oriented
to have maximum out-degree $1$.

\paragraph{Applying static guarantees to dynamic settings}
In order to apply static guarantees in a dynamic setting, we define
the notion of a sequence of insert/delete operations satisfying a
static guarantee.

For $\epsilon \in (0, 1)$ and for a hash-function pair
$h = (h_1, h_2)$, we say that a sequence
$\Psi = \langle \psi_1, \psi_2, \ldots\rangle$ of insert/delete
operations is \defn{$(\epsilon, h)$-viable with a stash of size $s$}
if the following holds: for every subsequence of operations of the
form
$P_i = \langle \psi_{i \epsilon n + 1}, \psi_{i \epsilon n + 2},
\ldots, \psi_{(i + 1) \epsilon n} \rangle$, the set $X$ of records
that are present (at any point) during the operations $P_i$ has the
property that $X$ is $h$-viable with a stash of size $s$.

The dynamic guarantees in this section will assume only that the
sequence of operations $\Psi$ is $(\epsilon, h)$-viable (with a stash
of size $s$) for some known parameter $\epsilon \in (0, 1)$, and will
make no other assumptions about $\Psi$ or the hash-function pair
$h = (h_1, h_2)$.

Note that the property of being $(\epsilon, h)$-viable is a statement
about the sets of records $X$ that are present during windows of
$\epsilon n$ operations. If the table is always filled to capacity
$cn$, for some $c \in (0, 1)$, then the property of being
$(\epsilon, h)$-viable is a statement about sets of $(c + \epsilon) n$
records. Thus dynamic guarantees for tables on $cn$ records can be
derived from static guarantees that apply to tables of
$(c + \epsilon) n$ records. By making $\epsilon$ smaller, one can close the gap between the capacities for the static and dynamic guarantees -- but as we shall see, this also increases the constant in the algorithm's running time.

\paragraph{Our dynamic guarantee}
Formally, we say that an \defn{implementation of a $k$-associative
  Cuckoo hash table with a stash of size $s$} is an algorithm that
maintains a Cuckoo hash table with $n$ bins, each of size $k$, and
with a stash of size up to $s$. The implementation is given two hash
functions $h_1, h_2$, and every record $x$ in the table must either be
stored in one of the bins $h_1(x), h_2(x)$ or in the stash. The
implementation is permitted to maintain an additional $O(n)$-space
data structure $\mathcal{D}$ for additional bookkeeping, as long as
$\mathcal{D}$ is not modified by queries, and as long as each
insert/delete incurs at most $O(1)$ writes to $\mathcal{D}$.

We say that a Cuckoo hash table implementation satisfies \defn{the
  dynamic guarantee} on a sequence of operations $\Psi$, if:
\begin{itemize}
\item Each insert/delete operation incurs $O(\log \log n)$ kickouts
  and takes time $O(\log n \log \log n)$.
\item The amortized cost of each insert/delete operation is $O(1)$.
\end{itemize}

The goal of this section will be to describe an implementation of
Cuckoo hashing that offers the dynamic guarantee (with high
probability) as long as the underlying sequence of operations $\Psi$ is
$(\epsilon, h)$-viable. We call our implementation of Cuckoo hashing
the \defn{Dancing-Kickout Algorithm}.

The main result of the section is the following theorem.

\begin{theorem}
  Let $\epsilon \in (0, 1)$ and $s$ be constants ($s$ may be $0$). Let
  $h = (h_1, h_2)$ be a pair of hash functions. Let $\Psi$ be a sequence
  of $\poly(n)$ insert/delete operations that is
  $(\epsilon, h)$-viable with a stash of size $s$.

  Then, with high probability in $n$, the Dancing-Kickout Algorithm
  implements an $8$-associative Cuckoo hash table with a stash of size
  $s$ that satisfies the dynamic guarantee on $\Psi$.
\end{theorem}
\begin{proof}
  We take the approach of starting with a weaker version of the
  theorem and then working our way towards the full version. Initially we will consider only
  inserts, but no deletes or stash. Then we will consider only inserts
  and a stash, but no deletes. Then we will consider all of inserts,
  deletes, and a stash, but we will make what we call the
  \defn{full-viability assumption}, which is that the set $X$ of
  \emph{all} of records inserted and deleted by $\Psi$ is
  $h$-viable. Finally, we will show how to remove the full-viability
  assumption, thereby obtaining the full theorem.
  
  We begin by describing the Dancing-Kickout Algorithm in the case
  where $\Psi$ consists of only insertions (and no deletions). In this
  case, the algorithm only uses the first $4$ slots in each bin. We
  also begin with the simplifying assumption that the stash size $s$
  is $0$.

  The algorithm thinks of each record $x$ as representing an edge
  $(h_1(x), h_2(x))$ in the Cuckoo graph $G$. Since the set of records
  $X$ being inserted is $h$-viable, it must be that $G$ can be
  oriented with out-degree $1$. This means that each connected
  component in $G$ is a pseudotree (i.e., a tree with up to one
  additional edge added).

  In this case, the Dancing-Kickout Algorithm works as
  follows. Whenever an edge-insertion connects two vertices from
  different connected components, the Dancing-Kickout Algorithm
  simply uses the Rank-Based Dancing-Walk Algorithm to maintain an
  edge-orientation with maximum out-degree $3$. On the other hand,
  when an edge-insertion connects two vertices $v, u$ that are already
  in the same tree as one another (we call the edge $(v, u)$ a
  \defn{bad edge}), the Dancing-Kickout Algorithm orients the edge
  arbitrarily and then disregards that edge in all steps (i.e., the
  edge cannot be used as part of a random walk). Since $G$ is a
  pseudoforest, each vertex $v$ is incident to at most one bad edge;
  it follows that the maximum out-degree in the graph never exceeds
  $4$. This, in turn, means that no bin in the Cuckoo hash table stores
  more than $4$ items.

  Lemma \ref{lem:rankcorrectness} and Theorem \ref{thm:amortized}
  ensure that the edge-insertions involving good edges satisfy the
  dynamic guarantee with high probability in $n$ (that is, each
  operation takes time $O(\log n \log \log n)$, incurs
  $O(\log \log n)$ edge flips, and takes amortized time $O(1)$). The
  edge-insertions involving bad edges can be analyzed as follows. Note
  that the time for the Rank-Based Dancing-Walk Algorithm to identify
  that an edge $e = (v, u)$ is bad is just the height of the rank tree
  containing $v$ and $u$. Since combination ranks never exceed
  $O(\log n)$, the time to identify a bad edge is never more than
  $O(\log n)$. Since each rank-tree will have at most one bad edge
  identified in it (because each connected component contains at most
  $1$ bad edge), the total time spent identifying bad edges is at most
  the sum of the depths of the rank trees (at the end of all edge
  insertions); this, in turn, is $O(n)$ since the depth of each rank
  tree is never more than the number of elements it contains. Thus the
  operations in which bad edges are inserted do not cause the dynamic
  guarantee to be broken.

  Now we describe what happens if $\Psi$ still consists only of
  insertions, but a stash of size $s > 0$ is used. In this case, the
  Dancing-Kickout Algorithm places an edge $e = (v, u)$ in the stash
  (i.e., the algorithm places the record $x$ for which $h_1(x) = v$
  and $h_2(x) = u$ in the stash) if $e$ is a bad edge \emph{and} if
  both of the vertices $v$ and $u$ are already incident to bad
  edges. On the other hand, if one of $v$ or $u$ is not already
  incident to a bad edge, then the edge can be oriented out-going from
  that vertex (just as was the case without a stash). Call an edge $e$
  \defn{super bad} if, when $e$ is inserted, there is already a bad
  edge in the connected component containing $e$. Since $\Psi$ is
  $h$-viable with a stash of size $s$, the number of super bad edges
  is at most $s$.\footnote{To see this formally, note that there must
    be a set of at most $s$ edges $Y$ such that $X \setminus Y$ is a
    pseudoforest. That is, without the edges $Y$ there would be
    \emph{no} super bad edges. On the other hand, one can verify that
    placing each of the edges from $Y$ back into the sequence of edges
    $X \setminus Y$ adds at most $|Y|$ super bad edges, since each
    edge that is placed in can increase the number of super bad edges
    by at most $1$.} Because the Random-Walk Algorithm only stashes
  super bad edges, the algorithm is guaranteed to never stash more
  than $s$ records at a time. The running time of the algorithm on
  non-super-bad edges is the same as in the case of no stash; on the
  other hand, the $s$ super bad edges can contribute
  $s \cdot O(\log n) =  O(\log n)$ in total to the running time of
  the algorithm. Thus, with high probability, the Random-Walk Algorithm
  still satisfies the dynamic guarantee.

  Now we consider what happens if $\Psi$ contains deletes in addition
  to inserts. To begin, consider the special case where the set $X$ of
  \emph{all records} that $\Psi$ \emph{ever} inserts (including those
  that are subsequently deleted) has the property that $X$ is
  $h$-viable -- we call this the \defn{full-viability
    assumption}. Under the full-viability assumption, deletes can be
  implemented with \defn{tombstones}, meaning that when a record is
  deleted it is simply marked as deleted without actually being
  removed. In fact, the use of tombstones is not actually
  necessary. This is because the analysis of the Rank-Based Dancing-Walk
  Algorithm for edge-orientation continues to work without
  modification even if edges in the graph disappear arbitrarily over
  time, as long as all of the edges (\emph{including those that
    disappear}) form a forest. Thus, in the case where the
  full-viability assumption holds, we can simply implement deletes by
  removing the appropriate record from the table, and then we can use
  the Dancing-Kickout algorithm exactly as described so far. Since
  the Rank-Based Dancing-Walk Algorithm can handle edges disappearing, it
  follows that the Dancing-Kickout algorithm still satisfies the
  dynamic guarantee with high probability.
  
  Finally, we consider what happens if $\Psi$ contains both inserts
  and deletes, but without making the full-viability assumption. So
  far, we have only used the first $4$ slots of each bin. We now
  incorporate into the algorithm slots $5, 6, 7, 8$, and we modify the
  algorithm to gradually rebuild the table in phases, where
  consecutive phases toggle between using only slots $1, 2, 3, 4$ or
  using only slots $5, 6, 7, 8$; as we shall see, each phase is
  individually designed so that the running-time of its operations can
  be treated as meeting the full-viability assumption.

  In more detail, the algorithm performs gradual rebuilds as
  follows. The operations $\Psi$ are broken into phases
  $P_1, P_2, \ldots$ each consisting of $\epsilon n$ operations. At
  the beginning of each phase $P_i$ where $i$ is even (resp. $i$ is
  odd), the hash table uses only the slots $1, 2, 3, 4$ (resp.
  $5, 6, 7, 8$) in each bin. During the phase of operations $P_i$, any
  new insertions are performed with the Dancing-Kickout Algorithm
  using slots $5, 6, 7, 8$ (resp. $1, 2, 3, 4$). Also, during the
  $j$-th operation in the phase $P_i$, the algorithm looks at bin $j$,
  takes any records in slots $1, 2, 3, 4$ (resp. $5, 6, 7, 8$), and
  reinserts those records into the hash table using slots $5, 6, 7, 8$
  (resp.  $1, 2, 3, 4$).\footnote{Additionally, if a stash of size
    $s > 0$ is used, then the first operation of each phase $P_i$
    reinserts all of the elements in the stash, using only slots
    $5, 6, 7, 8$ if $i$ is even and only slots $1, 2, 3, 4$ if $i$ is
    odd.} Finally, deletes are implemented simply by removing the
  appropriate record $x$, regardless of what slot that record may be
  in.

  During a given phase $P_i$, the algorithm can be thought of as
  starting with a new empty Cuckoo hash table (consisting in each bin
  of either the slots $1, 2, 3, 4$ if $i$ is odd or $5, 6, 7, 8$ if
  $i$ is even). Then over the course of $P_i$, one can think of the
  algorithm as performing not only the operations in $P_i$, but also
  populating the new hash table with any elements that were present at
  the beginning of the phase $P_i$ (unless those elements are deleted
  before they have a chance to be re-populated). Let $X$ be the set of
  all records $x$ that are placed into the new hash table at some
  point during $P_i$ (this includes both elements that operations in
  $P_i$ act on, as well as elements that are re-inserted due to the
  gradual rebuild during the phase).  By the $(\epsilon, h)$-viability
  of $\Psi$, we know that $X$ is $h$-viable. This means that phase
  $P_i$ can be analyzed as satisfying the full-viability
  assumption. Thus, with high probability in $n$, the algorithm does
  not violate the dynamic guarantee during phase $P_i$. Since there
  are $\poly(n)$ phases, it follows that, with high probability in
  $n$, the algorithm never violates the dynamic guarantee.
\end{proof}

\section{Proof of Claim \ref{clm:chernoffbound}}
\label{app:chernoffbound}
  
\begin{proof}[Proof of Claim \ref{clm:chernoffbound}]
  Consider any (possibly randomized) adaptive algorithm for Alice, and
  let $X$ be the random variable denoting Alice's profit in the game.

  For $\lambda > 0$, define
  $$M_{k, \mu}(\lambda) = \E[e^{\lambda X}]$$ to be the moment
  generating function of $X$. The key claim is that
  \begin{equation}
    M_{k, \mu}(\lambda) \le e^{(e^\lambda - 1)\mu}.
    \label{eq:momgen}
  \end{equation}

  We prove \eqref{eq:momgen} by induction on $k$. Suppose that
  \eqref{eq:momgen} holds for $M_{k', \mu}(\lambda')$ for all $k' <
  k$, any $\lambda' > 0$, and any adaptive algorithm for Alice; as a
  base case, \eqref{eq:momgen} is immediate for $k = 0$. Let $X_1$ be
  a random variable for the profit Alice makes from her first bet and
  $X' = X - X_1$. For any value $p \le \mu$ that Alice may select for
  $p_1$,
  $$\E\Big[e^{\lambda X} \mid p_1 = p\Big] = p e^\lambda \cdot
  \E\Big[e^{\lambda X'} \mid X_1 = 1, \ p_1 = p\Big] + (1 - p) \cdot
  \E\Big[e^{\lambda X'} \mid X_1 = 0, \ p_1 = p\Big].$$ By the
  inductive hypothesis,
  $$\E\Big[e^{\lambda X'} \mid X_1 = 1, \ p_1 = p\Big],
  \ \ \ \E\Big[e^{\lambda X'} \mid X_1 = 0, \ p_1 = p\Big] \le
  e^{(e^{\lambda} - 1)(\mu - p)}.$$ Thus
  $$\E\Big[e^{\lambda x} \mid p_1 = p\Big] \le \left(p \cdot e^\lambda
  + (1 - p)\right) e^{(e^{\lambda} - 1)(\mu - p)}.$$ Using the
  identity, $1 + x \le e^x$ with $x = p \cdot (e^{\lambda} - 1)$, it
  follows that
  $$\E\Big[e^{\lambda x} \mid p_1 = p\Big] \le e^{p (e^\lambda - 1)}
  e^{(e^{\lambda} - 1)(\mu - p)} = e^{(e^{\lambda} - 1) \mu}.$$ Since this
  holds for all $p$, \eqref{eq:momgen} follows.
  
  Using \eqref{eq:momgen}, we can complete the proof of the lemma as
  follows. By Markov's inequality,
  $$\Pr\Big[X > (1 + \delta) \mu\Big] \le \Pr\Big[e^{\lambda X} >
    e^{\lambda (1 + \delta) \mu}\Big] \le \frac{\E\Big[e^{\lambda
        X}\Big]}{e^{\lambda (1 + \delta) \mu}}.$$ By
  \eqref{eq:momgen}, it follows that
  $$\Pr\Big[X > (1 + \delta) \mu\Big] \le \exp\left( \left(e^{\lambda}
  - 1 - \lambda (1 + \delta)\right) \mu\right).$$ Plugging in $\lambda
  = \ln(1 + \delta)$, which one can show by the derivative test
  minimizes the expression on the right side, yields
  $$\Pr\Big[X > (1 + \delta) \mu\Big] \le \exp\left((\delta - \ln(1 +
    \delta) (1 + \delta))\mu\right).$$
\end{proof}

\end{document}